\title{Reputation Building under Observational Learning}
\author{Harry PEI\footnote{Department of Economics, Northwestern University. I thank Daron Acemoglu, S. Nageeb Ali, Alp Atakan, Jie Bai, Dhruva Bhaskar, Drew Fudenberg, Olivier Gossner, Johannes Johnen, George Lukyanov, Wojciech Olszewski, Peter Norman S{\o}rensen,
Asher Wolinsky, Alex Wolitzky, and my seminar participants for helpful comments. Errors are mine.}}
\date{June 27, 2020}

\documentclass[11pt]{article}
\usepackage{amsmath}
\usepackage{graphicx}
\usepackage{amsfonts}
\usepackage{amsthm}
\usepackage{amssymb}
\usepackage{setspace}
\usepackage{tikz}
\usetikzlibrary{patterns}
\usepackage{sgame}
\usepackage{color}
\usepackage{hyperref}
\usepackage[top=0.9in, bottom=0.9in, left=0.9in, right=0.9in]{geometry}
\begin{document}
\newtheorem{Proposition}{\hskip\parindent\bf{Proposition}}
\newtheorem{Theorem}{\hskip\parindent\bf{Theorem}}
\newtheorem*{Theorem1}{\hskip\parindent\bf{Theorem 1'}}
\newtheorem*{Theorem3}{\hskip\parindent\bf{Theorem 3'}}
\newtheorem*{Theorem4}{\hskip\parindent\bf{Theorem 4'}}
\newtheorem{Lemma}{\hskip\parindent\bf{Lemma}}[section]
\newtheorem{Corollary}{\hskip\parindent\bf{Corollary}}
\newtheorem{Definition}{\hskip\parindent\bf{Definition}}
\newtheorem{Assumption}{\hskip\parindent\bf{Assumption}}
\newtheorem{Condition}{\hskip\parindent\bf{Condition}}
\newtheorem*{Conjecture2}{\hskip\parindent\bf{Complete Records Benchmark}}
\newtheorem*{Conjecture1}{\hskip\parindent\bf{Modeling Assumption}}
\newtheorem*{Conjecture3}{\hskip\parindent\bf{Implications of Theorems 1 and 2}}
\newtheorem*{Conjecture4}{\hskip\parindent\bf{Implication of Theorem 4}}
\newtheorem{Conjecture}{\hskip\parindent\bf{Implication}}
\maketitle
\numberwithin{equation}{section}

\noindent
I study a social learning model in which the object to learn is a strategic player's \textit{endogenous actions} rather than an \textit{exogenous state}. A patient seller faces a sequence of buyers and decides whether to build a reputation for supplying high quality products. Each buyer does not have access to the seller's complete records, but can observe all previous buyers' actions, and some informative private signal about the seller's actions. I examine how the buyers' private signals affect the speed of social learning and the seller's incentives to establish reputations. When each buyer privately observes a bounded subset of the seller's past actions, the speed of learning is strictly positive but can vanish to zero as the seller becomes patient. As a result, reputation building can lead to low payoff for the patient seller and low social welfare. When each buyer observes an unboundedly informative private signal about the seller's current-period action, the speed of learning is uniformly bounded from below and a patient seller can secure high returns from building reputations. My results shed light on the effectiveness of various policies in accelerating social learning and encouraging sellers to establish good reputations.\\


\begin{spacing}{1.5}
\section{Introduction}\label{sec1}
Recent empirical findings suggest that reputation mechanisms break down in a variety of markets in developing countries.
For example, in the
 markets for malaria drugs (Nyqvist, et al. 2018), fruits (Bai 2018), and milk powder (Bai, et al. 2019),
consumers believe that sellers are likely to supply low quality.
These pessimistic beliefs persist over time, which lower
the seller's returns from supplying high quality, and make these beliefs self-fulfilling.
These ``bad equilibria'' are at odds with the canonical reputation results in Fudenberg and Levine (1989,1992), which suggest that buyers' mistrust cannot persist and patient sellers can secure high returns from a reputation for supplying high quality.

This paper presents a reputation model, which, among other results, suggests a rationale for such reputation failures and persistent mistrust. I argue that when each consumer has \textit{limited access to a seller's past records} (e.g., observes
a bounded subset of the seller's past actions), and learns primarily from \textit{previous consumers' choices}, consumers' learning can be arbitrarily slow, which
can wipe out the patient seller's returns from building reputations. My results also suggest
policy interventions that can accelerate learning and can restore seller's incentives to build reputations.


My modeling assumptions, limited access to sellers' past records and learning from others' choices, fit into a number of retail markets
in developing economies. Limited availability of formal records may result from inadequate record-keeping technologies.
Even when official records are available, their credibility is undermined by institutional failures such as collusion between merchants and bureaucrats. As a result, information about the seller's past actions is dispersed among consumers
(e.g., each consumer observes the quality of product she bought from the seller), and is passed on to future buyers via consumers' choices and
word-of-mouth communication.\footnote{Evidence for observational learning and word-of-mouth communication
is documented in
the markets for food (Cai, Chen and Fang 2009), fertilizers (Conley and Udry 2010), drugs (Nyqvist, et al. 2018), and so on.}

Can consumers' social learning provide adequate incentives for sellers to supply high quality? How does
the speed of learning depend on
consumers' private signals?
For example, does a seller have stronger reputational incentives when each consumer observes his actions in the last $10$ periods
(i.e., after selling a low quality product, the seller will be punished by the next 10 consumers), or when each consumer observes some private signal about his current-period action? Should regulators issue quality certificates to high-quality products or inform consumers about low-quality products? Answering these questions
not only helps to understand the causes of reputation failures, but also
informs policy-makers about which information to provide to consumers in order to restore efficiency.\footnote{Some obvious solutions to the  reputation failure problem include: inform consumers about the seller's complete records, or screen out all low-quality products sold on the market. However, these policies are hard to implement in practice due to their formidable costs, and their violations of consumers' capacity constraints to process detailed information.
Policies that respect budget and attention constraints include: (1) marginal improvements in record-keeping technologies that
allow each consumer to observe
a longer history of the seller's past actions; (2) inspect a \textit{small fraction} of products currently sold on the market and
inform consumers about the quality of inspected products.}

Motivated by these questions, I study  an infinitely repeated game between a patient player $1$ (e.g., seller) and a sequence of player $2$s (e.g., consumers), arriving one in each period and each plays the game only once.
Player $1$ is either a strategic type who maximizes his discounted average payoff, or a commitment type who plays his
(pure) Stackelberg action in every period. I focus on situations in which the commitment type occurs with small but positive probability.

My modeling innovation is in the monitoring structure. In my baseline model, every player $2$ observes the \textit{entire history} of player $2$s' actions, in addition to player $1$'s actions in a \textit{bounded number of periods}. This resembles retail markets in developing countries. Given the geographic proximity, consumers can casually observe each other's choices in these localized markets. However, learning about the seller's actions (i.e., the quality and attributes of various products he sold) requires more time and effort. For example, a consumer needs to talk to his friends and learn about their personal experiences.
Due to the time costs
of these conversations  (Niehaus 2009) and consumers' limited capacity to process information (Sims 2003), it is reasonable to assume that each consumer only has a limited number of friends who are willing to spend time answering her questions, and as a result, can
learn at most a bounded number of the seller's past actions before making her decision.\footnote{Evidence for consumers' limited attention or limited capacity to process detailed information is abundant in the empirical IO literature. See for example, Heidhues, Johnen and Koszegi (2020) and the references therein.}

When player $2$ observes player $1$'s actions in the last $K$ periods, Theorem \ref{Theorem1} shows that there exist equilibria in which player $1$'s payoff is no more than his \textit{worst stage-game equilibrium payoff}.\footnote{Theorem \ref{Theorem1} applies to every game that has a pure strategy Nash Equilibrium and satisfies a generic assumption on stage-game payoffs. Appendix A.3 provides conditions under which player $1$ attains his \textit{minmax payoff} in equilibrium.} This is the case no matter how patient he is and how large $K$ is. In \textit{monotone-supermodular} games that fit into buyer-seller applications, there exist equilibria in which both players receive their respective minmax payoffs (Theorem \ref{Theorem2}).
My reputation failure result extends when each
player $2$ \textit{randomly samples}
a \textit{bounded subset} of her predecessors,
observes player $1$'s actions against the individuals in her sample, in addition to
all her predecessors' actions (Theorem \ref{Theorem3}). My findings contrast to canonical reputation results
which suggest that a patient player can secure his Stackelberg payoff in \textit{all equilibria}.

My proof constructs a class of equilibria in which the speed of learning is \textit{strictly positive},
but \textit{vanishes to zero} as player $1$ becomes arbitrarily patient. To see how slow learning causes reputation failures, take $K=1$ and consider the product choice game in Mailath and Samuelson (2001):
\begin{center}
\begin{tabular}{| c | c | c |}
  \hline
  -- & $T$ & $N$ \\
  \hline
  $H$ & $2,1$ & $-1,0$ \\
  \hline
  $L$ & $3,-1$ & $0,0$ \\
  \hline
\end{tabular}
\end{center}
Play starts from a \textit{non-trusting phase} in which $(L,N)$ is played with high probability, and enters a \textit{trusting phase} as soon as some buyer has played $T$, after which $(H,T)$ is played on the equilibrium path. In retail markets of developing countries, these phases translate into two self-fulfilling social norms: one in which buyers do not trust the seller and the seller supplies high quality with low probability, another one in which buyers trust the seller and the seller supplies high quality with high probability. The probability that a buyer plays $T$ in the non-trusting phase depends on the seller's action in the period before, which is zero if the seller supplied low quality, and is strictly positive if the seller supplied high quality (which I call \textit{the rate of trust building}).

The rate of trust building decreases with the seller's patience since the seller's incentive to exert high effort is to affect the buyers' actions in the next period (which can be observed by all future buyers and affect the seller's continuation value). When the seller is patient, he
is willing to play $H$ despite his action affects the buyer's future action with low probability. When the buyer's action becomes less responsive to the seller's action in the previous period, it
lowers the speed with which future buyers learn about the seller's type, and
 prolongs process of building reputations.

In order to confirm that the driving force behind Theorem \ref{Theorem1} is \textit{slow learning} rather than alternative mechanisms proposed in the social learning and reputation literature,
I establish three properties that apply to \textit{all equilibria}. From the perspective of social learning,
these properties highlight the new economic forces when the learning process is
endogenously controlled by a strategic long-run player.

First, player $2$s never herd on any action that does not best reply against player $1$'s Stackelberg action. Intuitively, player $1$ has no intertemporal incentive after player $2$s herd, and hence, plays a myopic best reply against that herding action.
When player $1$'s myopic best reply is his Stackelberg action, all player $2$s best reply against this Stackelberg action.
When player $1$'s best reply is not his Stackelberg action, player $2$ believes that player $1$ is the commitment type after observing him playing the Stackelberg action, and therefore, best replies against the Stackelberg action in the next period. Both conclusions contradict the presumption on player $2$s' herding.

Using similar ideas, I show that when players' stage-game payoffs are monotone-supermodular, (1) player $2$s' actions in the next $K$ periods are informative about player $1$'s current-period action unless
player $1$ is guaranteed to receive his optimal commitment payoff in the next $K$ periods; (2) player $1$'s asymptotic payoff from establishing a reputation is at least a fraction $\frac{K}{K+1}$ of his optimal commitment payoff. These findings suggest that
Theorem \ref{Theorem1} is not driven by low-payoff outcomes in the long run,
or player $2$s' actions being uninformative, which stand in contrast to models of bad reputations.

Motivated by  policy interventions such as a regulator randomly inspects a small fraction of products sold on the market and informs consumers about the quality of inspected products,
Section \ref{sec5} examines an alternative specification of player $2$s' private information, in which every player $2$ observes
a \textit{private signal} of player $1$'s \textit{current period action} (e.g., the inspect result or the absence thereof),
in addition to previous player $2$s' actions, and possibly, player $1$'s actions in the last $K \in \mathbb{N} \cup\{0\}$ periods.

Theorem \ref{Theorem4} shows that in games where player $1$'s action choice is binary, he can secure his commitment payoff in all equilibria \textit{if and only if} player $2$'s private signal is \textit{unboundedly informative}, namely, there exists a signal realization that occurs with positive probability only when player $1$ plays his Stackelberg action.
In games where player $1$ has three or more actions, Theorem \ref{Theorem4} extends when players' stage-game payoffs are monotone-supermodular and
the distribution of private signal satisfies a \textit{monotone likelihood ratio property} (or MLRP), which means that the likelihood ratio between a high signal realization and a low signal realization increases when player $1$ takes a higher action.\footnote{I show by counterexample that the monotone-supermodularity and MLRP requirements are not redundant, i.e., when player $1$ has three or more actions, player $2$'s private signal being  unboundedly informative is \textit{neither necessary nor sufficient} for player $1$ to secure high returns from building reputations when one of these conditions is violated.}

Theorem \ref{Theorem4} is reminiscent of a result in Smith and S{\o}rensen (2000), that myopic agents' actions asymptotically match the state if and only if their private signals are unboundedly informative. However, their result \textit{does not} imply that player $1$ can secure a high payoff for two reasons. First, the myopic players in my model learn about the \textit{endogenous actions} of a strategic player rather than an \textit{exogenous state}. Second, converging to a high-payoff outcome asymptotically \textit{does not} imply that a patient patient receives a high discounted average payoff. What also matters is the rate of convergence.

I establish Theorem \ref{Theorem4} by showing that in binary action games (or monotone-supermodular games with MLRP), player $2$s observing an unboundedly informative private signal guarantees a \textit{lower bound on the speed of learning}, which is \textit{independent of player $1$'s discount factor}. In particular, after observing previous player $2$s' actions but before observing her private signal, as long as player $2$ believes that she will \textit{not} best reply against the Stackelberg action with positive probability, the probability that she best replies against the Stackelberg action must be \textit{strictly higher} when player $1$ plays the Stackelberg action.
As a result, the informativeness of player $2$'s action
is uniformly bounded away from zero, and a patient player $1$ can secure his Stackelberg payoff
 by building his reputation.

In terms of policy implications, Theorems \ref{Theorem1}, \ref{Theorem2}, and \ref{Theorem3} suggest that in markets with significant adverse selection, lack of formal records, and consumers relying on observational learning,
marginal improvements in the record-keeping technology (i.e., increase $K$ to another finite number, or remove the noise in the observation of  seller's past actions) is ineffective:
the market may remain in a bad equilibrium in which learning is slow, the seller's reputational incentives are weak, and mistrust between buyers and sellers persist for a long time.

Applying to monotone-supermodular games (which include but not limited to the product choice game), Theorem \ref{Theorem4} suggests that randomly inspecting a \textit{small fraction of products} currently sold on the market and issuing quality certificates to high-quality products can restore patient seller's incentives to supply high quality. Due to the presence of observational learning,
this policy is effective
even when the certificate is noticed only by consumers who demand products in the current period.
However, the regulator needs to ensure that the certificate issued to high-quality products \textit{cannot be forged} and cannot be used on low-quality products, i.e., the quality certificate needs to be unboundedly informative about the seller's Stackelberg action.
By contrast, informing consumers only about which of the inspected products have low quality is ineffective.
The details are discussed in Section \ref{sec6}.

\paragraph{Related Literature:} This paper contributes to the social learning literature by highlighting the new economic forces when
the \textit{object to learn} is a strategic long-run player's \textit{endogenous actions} rather than an \textit{exogenous state}.
Due to the long-run player's incentive constraints and the presence of  commitment type, the myopic players \textit{never} herd on any action that does not best reply against the commitment action, and the bad equilibrium is driven by the low rate of learning.
This contrasts to the canonical social learning models of Banerjee (1992), Bikhchandani, Hirshleifer and Welch (1992) and Smith and S{\o}rensen (2000), in which inefficiencies are driven by players herding on the wrong action.\footnote{Players also learn about an exogenous state in the recent works of Rosenberg and Vieille (2019) and Harel, Mossel, Strack and Tamuz (2020). Bose, Orosel, Ottaviani and Vesterlund (2006) and Arieli, Koren and Smorodinsky (2017) provide necessary and sufficient conditions for herding to occur when the state is exogenous but the price for each action is endogenous.  In Kultti and Miettinen (2006), Mueller-Frank and Pai (2016), and Ali (2018),
the myopic players learn about an exogenous state, but can endogenously decide how many predecessors' actions to observe. A contemporary work by Logina, Lukyanov and Shamruk (2019) studies a social learning model in which every myopic player observes a private signal about a patient player's action. They show that the patient player exerts high effort when the myopic players' beliefs are intermediate and exerts low effort otherwise. The intuition is similar to career concern models, in which a patient player's reputational incentives are weaker when the public has more precise belief about his type. }

In terms of research question, I examine the effects of social learning on a patient player's \textit{discounted average payoff}. This contrasts to existing results that focus on players' asymptotic beliefs,  asymptotic rates of learning (Gale and Kariv 2003, Harel, Mossel, Strack and Tamuz 2020), and asymptotic payoffs (Rosenberg and Vieille 2019).\footnote{A separate strand of literature characterizes the optimal mechanism that maximizes myopic players' discounted average payoff in social learning models where these players move sequentially and learn about an \textit{exogenous state}. Some notable examples include Che and H\"{o}rner (2018) and
Smith, S{\o}rensen and Tian (2020).}
My results relate the patient player's guaranteed equilibrium payoff to
the \textit{minimal rate of learning} that satisfies his incentive constraints, and in particular, whether
this rate is bounded away from zero when player $1$ is arbitrarily patient. This perspective unifies my reputation failure results (Theorems \ref{Theorem1} to \ref{Theorem3}) and my positive reputation result (Theorem \ref{Theorem4}).

In the constructive proofs of my reputation failure results, the patient player's discounted average payoff is low despite his asymptotic payoff is high. This comparison highlights the distinction between \textit{discounted average payoff} and \textit{asymptotic payoff} in models where the rate of learning is endogenous.

My paper contributes to the reputation literature by establishing a reputation result when myopic players learn from each other's actions in addition to observing some informative private signals (Theorem \ref{Theorem4}). My model captures consumers' learning
in markets with asymmetric information, lack of formal records, and consumers relying on observational learning and word-of-mouth communication.
To the best of my knowledge, this has not been explored in the existing reputation literature.

My Theorems \ref{Theorem1} to \ref{Theorem3} identify a new mechanism that accounts for reputation failures in these markets. In particular, reputation fails since the speed of learning vanishes as the reputation-building player becomes patient. This differs from existing theories that are based on the uninformed player's forward-looking incentives and the lack-of identification of the informed player's actions.

In particular,
models with lack-of identification such as Ely and V\"{a}lim\"{a}ki (2003), Ely, Fudenberg and Levine (2008), and Deb, Mitchell and Pai (2020) focus on \textit{participation games}, in which the uninformed player(s) can take a \textit{non-participating action} under which the public signal is uninformative about the informed player's current period action. In Levine (2019), the signals are less informative when
 the uninformed players do not participate.
These features contrast to my model in which
the uninformed players' actions \textit{cannot} prevent her successors from observing the informed player's current period action, and
their actions are informative about the informed player's actions in the past.

Cripps and Thomas (1997) and Chan (2000) construct low-payoff equilibria when players are equally patient and the uninformed player can observe the entire history of the informed player's actions. In their models as well as other models with complete records,
\textit{uninformed player's patience} undermines reputation building while
the informed player's patience helps reputation building.
In my model, the uninformed players are myopic, but the informed player's patience decreases the rate of learning and causes reputations to fail.
I provide a more detailed explanation
in Section \ref{sub7.2}.

My paper is also related to repeated games with limited records,\footnote{See for example, the recent work of Bhaskar and Thomas (2019) and the references therein.} and more specifically, reputation models with bounded memories, which include
Liu (2011), Liu and Skrzypacz (2014), and Kaya and Roy (2020).
These papers study reputation games in which every short-run player observes a bounded sequence of the long-run player's past actions, but
\textit{cannot} observe each other's actions. Their assumption rules out the possibility of observational learning. By contrast,
my paper studies the effects of observational learning on a patient player's reputational incentives, as well as policy interventions that can accelerate observational learning and provide adequate incentives to build reputations.

\section{Baseline Model}\label{sec2}
Time is discrete, indexed by $t=0,1,2...$. A long-lived player $1$ (he, e.g., seller) with discount factor $\delta \in (0,1)$ interacts with an infinite sequence of short-lived player $2$s (she, e.g., buyers), arriving one in each period and each plays the game only once. In period $t$, players simultaneously choose their actions $a_t$ and $b_t$ from finite sets $A$ and $B$.
Players have access to a public randomization device.
Let $\xi_t$ be its realization in period $t$,
which is uniformly distributed on $[0,1]$.

Players' stage-game payoffs are $u_1(a_t,b_t)$ and $u_2(a_t,b_t)$. Let $\textrm{BR}_1: \Delta (B) \rightrightarrows 2^{A} \backslash \{\varnothing\}$ and $\textrm{BR}_2: \Delta (A) \rightrightarrows 2^{B} \backslash \{\varnothing\}$ be player $1$'s and player $2$'s best reply correspondences in the stage-game. The set of player $1$'s (pure) Stackelberg actions is
$\arg\max_{a \in A} \big\{ \min_{b \in \textrm{BR}_2(a)} u_1 (a,b) \big\}$.
I introduce two assumptions on players' stage-game payoffs:
\begin{Assumption}\label{Ass1}
$\textrm{BR}_1(b)$ is a singleton for every $b \in B$.
$\textrm{BR}_2(a)$ is a singleton for every $a \in A$.
Player $1$ has a unique pure Stackelberg action.
\end{Assumption}
A sufficient condition for
Assumption \ref{Ass1} is that $u_i(a,b) \neq u_i(a',b')$
for every $i \in \{1,2\}$ and $(a,b) \neq (a',b')$. This is satisfied for generic $(u_1,u_2)$ given that $A$ and $B$ are finite sets.
Let $a^*$ be player $1$'s (pure) Stackelberg action.
Let $b^* \in B$ be the unique element in $\textrm{BR}_2(a^*)$, which I refer to as player $2$'s \textit{Stackelberg best reply}. Let
$u_1(a^*,b^*)$ be player $1$'s \textit{Stackelberg payoff}.
\begin{Assumption}\label{Ass2}
There exists a pure-strategy Nash equilibrium in the stage-game.
\end{Assumption}
Under Assumption \ref{Ass2}, player $1$'s pure Stackelberg payoff is weakly greater than his payoff in any pure-strategy Nash equilibrium.
This assumption is satisfied in most of the games studied in the reputation literature, such as product choice games, chain store games, and coordination games. It rules out games such as rock-paper-scissors in which commitment to pure actions is not beneficial.

Player $1$ has perfectly persistent private information about his \textit{type}
$\omega \in \{\omega^s, \omega^c\}$, where $\omega^c$ stands for a \textit{commitment type} who mechanically plays the Stackelberg action $a^*$ in every period, and $\omega^s$ stands for a \textit{strategic type} who  can flexibly choose his actions in order to maximize his payoff.

Player $1$ can observe the all the actions taken in the past in addition to the current and past realizations of public randomization devices.
Let $h_1^t$ be a typical private history of the strategic-type player $1$ in period $t$, with $h_1^t \equiv \{ a_0,...,a_{t-1},b_0,...,b_{t-1},\xi_0,...,\xi_t\}$.
Let $\mathcal{H}_1^t$ be the set of $h_1^t$ and let $\mathcal{H}_1 \equiv \cup_{t=0}^{\infty} \mathcal{H}_1^t$.
Strategic-type player $1$'s strategy is
$\sigma_1 : \mathcal{H}_1 \rightarrow \Delta (A)$, with $\sigma_1 \in \Sigma_1$.

Player $2$'s prior belief attaches probability $\pi_0 \in (0,1)$ to the commitment type $\omega^c$. Her private history coincides with the public history, which consists of calendar time, all her predecessors' actions (obtained from observational learning), player $1$'s actions in the past $K \in \{1,2,3,...\}$ periods (obtained from word-of-mouth communication with player $2$s who arrive before her), and the current realization of public randomization device.\footnote{In Section \ref{sec4}, I generalize my result to settings in which each player $2$ observes a bounded \textit{stochastic} subset of player $1$'s past actions. In Section \ref{sec5}, I examine the game's outcomes when each player $2$ can also observe an informative signal about player $1$'s current-period action. The public randomization device is introduced to ease the exposition. My results also apply under alternative assumptions on the observability of the public randomization device.}
The exogenous parameter $K$ measures player $2$'s capacity to process detailed information about player $1$'s past actions
(e.g., quality and attributes of seller's products).

Formally,
let $h^t$ be the public history in period $t$, with $h^t \equiv \{b_0,b_1,...,b_{t-1},a_{\max\{0,t-K\}},...,a_{t-1},\xi_{t}\}$.
Let $\mathcal{H}^t$ be the set of $h^t$ and let $\mathcal{H} \equiv \cup_{t=0}^{\infty} \mathcal{H}^t$.
Player $2$'s strategy is $\sigma_2 : \mathcal{H} \rightarrow \Delta (B)$, with $\sigma_2 \in \Sigma_2$.
Let $\pi(h^t)$ be the probability
that player $2$'s belief at $h^t$ attaches to the commitment type,
which I refer to as player $1$'s \textit{reputation} at $h^t$.

For every strategy profile $(\sigma_1,\sigma_2)$, strategic-type player $1$'s \textit{discounted average payoff} is
\begin{equation}
\mathbb{E}_1^{(\sigma_1,\sigma_2)} \Big[   \sum_{t=0}^{\infty} (1-\delta) \delta^t u_1(a_t,b_t) \Big],
\end{equation}
where $\mathbb{E}_1^{(\sigma_1,\sigma_2)}[\cdot]$ is the expectation over histories
when player $2$s play according to $\sigma_2$ and player $1$ plays according to $\sigma_1$.
Player $2$s' (discounted average) welfare is
\begin{equation}
\mathbb{E}^{(\sigma_1,\sigma_2,\pi_0)} \Big[   \sum_{t=0}^{\infty} (1-\delta_s) \delta_s^t u_2(a_t,b_t) \Big],
\end{equation}
where $\mathbb{E}^{(\sigma_1,\sigma_2,\pi_0)}[\cdot]$ is the expectation over histories
when player $2$s play according to $\sigma_2$, player $1$ plays according to $\sigma_1$ with probability $1-\pi_0$, and plays $a_1^*$ in every period with probability $\pi_0$,
and $\delta_s \in (0,1)$ is a discount factor that a planner uses to evaluate player $2$s' welfare and can be different from $\delta$.

I use sequential equilibrium for results on reputation failures,\footnote{I adopt the definition of sequential equilibrium in P\c{e}ski (2014, page 658) for this infinite horizon game.} i.e., the existence of equilibria in which a patient player $1$ receives a low payoff. This ensures that the equilibria I construct are not driven by uninformed players' unreasonable off-path beliefs.
I use Bayes Nash equilibrium (or BNE)
for results that establish the common properties of all equilibria. This ensures the robustness of my findings against equilibrium selection.

\section{Reputation Failure under Social Learning}\label{sec3}
I show that reputation building can result in low payoffs for both players, which contrasts to
the conclusions of canonical reputation models.
Let $(a',b') \in A \times B$ be player $1$'s \textit{worst} pure-strategy Nash Equilibrium in the stage-game, which exists under Assumption \ref{Ass2}.
Let $\underline{v}_1 \equiv u_1(a',b')$, which by definition, is weakly lower than $u_1(a^*,b^*)$.
Let
\begin{equation*}
\underline{\delta} \equiv \left\{ \begin{array}{ll}
\max \Big\{
\frac{\max_{a \in A} u_1(a,b^*)-u_1(a^*,b^*)}{\max_{a \in A} u_1(a,b^*)-\underline{v}_1} , \frac{\underline{v}_1-u_1(a^*,b')}{u_1(a^*,b^*)-u_1(a^*,b')}
\Big\}  & \textrm{ if } \underline{v}_1 < u_1(a^*,b^*) \\
0 & \textrm{ if } \underline{v}_1 = u_1(a^*,b^*).
\end{array} \right.
\end{equation*}
For a given
parameter configuration $(\delta,\pi_0,K)$, let $\textrm{SE}(\delta,\pi_0,K) \subset \Sigma_1 \times \Sigma_2$ be the set of strategy profiles
that are part of some sequential equilibria.

\begin{Theorem}\label{Theorem1}
When stage-game payoffs satisfy Assumptions \ref{Ass1} and \ref{Ass2}.
For every $K \in \mathbb{N}$, there exists $\overline{\pi}_0 \in (0,1)$,
such that for every $\pi_0 \in (0,\overline{\pi}_0)$ and $\delta \geq \underline{\delta}$,
there exists $(\sigma_1^{\delta},\sigma_2^{\delta}) \in \textrm{SE}(\delta,\pi_0,K)$, such that:
\begin{equation}\label{3.1}
     \mathbb{E}_1^{(\sigma_1^{\delta},\sigma_2^{\delta})} \Big[   \sum_{t=0}^{\infty} (1-\delta) \delta^t u_1(a_t,b_t) \Big] = \underline{v}_1.
\end{equation}
\end{Theorem}
According to Theorem \ref{Theorem1}, when the prior probability of commitment type is low,
there exist sequential equilibria in which the long-run player receives his lowest stage-game equilibrium payoff. This applies
regardless of his discount factor $\delta$ and player 2s' capacity to process information $K$.


Theorem \ref{Theorem1} demonstrates the failure of reputation effects when player $1$'s Stackelberg payoff is strictly greater than $u_1(a',b')$, i.e., player $1$ can strictly benefit from commitment in the stage game:
\begin{Condition}[Strict Benefit from Commitment]
$u_1(a^*,b^*)>\underline{v}_1$.\footnote{Condition 1 is less demanding than the lack-of-commitment condition in Cripps, Mailath, and Sameulson (2004, Assumption 3), which requires that the Stackelberg action $a^*$ does not best reply against $b^*$. My \textit{strict benefit from commitment condition} is satisfied not only in product choice games and entry deterrence games, but is also satisfied in games with conflicting interests such as chicken games, and coordination games in which player $1$ receives different payoffs from different pure strategy equilibria.}
\end{Condition}
Under the strict benefit from commitment condition, Theorem \ref{Theorem1} stands in contrast to the conclusions in Fudenberg and Levine (1989, 1992) and Gossner (2011): if player $2$s have \textit{unbounded observations} of player $1$'s past actions (i.e., $K=\infty$), or more generally, unbounded observations of noisy signals that can statistically identify player $1$'s actions, then a patient player $1$ can \textit{guarantee} his Stackelberg payoff $u_1(a^*,b^*)$ in \textit{all Bayes Nash equilibria} of the reputation game.

I argue that reputation fails because player $2$s' learning is too slow. In particular, although player $2$s' actions are informative about player $1$'s past actions, their informativeness vanishes to zero as player $1$ becomes patient.
In order to rule out alternative mechanisms, I establish some common properties of all Bayes Nash equilibria in Section \ref{sub3.3}. I show that
when player $1$ imitates the commitment type,
player $2$s never herd on actions other than $b^*$, their actions are informative about player $1$'s actions, and player $1$'s asymptotic payoff is at least a fraction $\frac{K}{K+1}$ of his Stackelberg payoff.

The logic of slow learning is reflected in my constructive proof of Theorem \ref{Theorem1} in Section \ref{sub3.4},
and is explained via
 Gossner (2011)'s entropy approach in Section \ref{sub7.1}.
For an informal illustration of the construction, the equilibrium play consists of a \textit{non-trusting phase} in which $(a',b')$ is played with high probability, and a \textit{trusting phase} in which $(a^*,b^*)$ is played for sure. Play starts from the non-trusting phase, and enters the trusting phase as soon as some player $2$ has played $b^*$. The probability with which player $2$ plays $b^*$ in the non-trusting phase
is zero if  $a'$ was played in the period before, and is a strictly positive number $r$ if $a^*$ was played in the period before. The rate of phase transition $r$ measures the informativeness of player $2$'s action about player $1$'s action,
which is also the speed with which future player $2$s learn about player $1$'s type.

In the above equilibrium, player $2$'s action converges to $b^*$ with probability $1$, but the rate of convergence $r$ vanishes to $0$ as $\delta$ goes to $1$. This low rate of learning wipes out player $1$'s returns from building a good reputation, discourages him from investing in his reputation, and makes player $2$s' pessimistic beliefs about player $1$'s actions in the early stages of the game self-fulfilling.

Two natural questions follow from Theorem \ref{Theorem1}.
First, even when $\underline{v}_1$ is strictly lower than $u_1(a^*,b^*)$, it may not equal to player $1$'s minmax payoff. Second, Theorem \ref{Theorem1} demonstrates the negative payoff consequences from the long-run player's perspective, but a more important question in many applications is the short-run players' welfare (e.g., consumers' welfare in the buyer-seller application). In order to address these concerns:
\begin{enumerate}
\item I identify a class of games (i.e., games with monotone-supermodular payoffs) that fit into buyer-seller applications,
in which $\underline{v}_1$ equals player $1$'s minmax payoff.
In Appendix A.3, I provide sufficient conditions under which player $1$'s lowest equilibrium payoff in the reputation game equals his minmax payoff, which can be strictly lower than $\underline{v}_1$.
\item I show that slow observational learning also results in low welfare for player $2$s,
regardless of the discount factor a social planner uses to evaluate different generations of player $2$s' payoffs.
In games with monotone-supermodular payoffs, there exist equilibria in which both players attain their respective minmax payoffs (Theorem \ref{Theorem2}).
\end{enumerate}

\paragraph{Attaining Long-Run Player's Minmax Payoff:} I review the notion of player $1$'s minmax payoff in Fudenberg, Kreps and Maskin (1990), which
takes player $2$s' myopia into account.
Let
\begin{equation}
    \mathcal{B}^* \equiv \{\beta \in \Delta(B) | \exists \alpha \in \Delta(A) \textrm{ s.t. } \beta \in \textrm{BR}_2(\alpha)\} \subset \Delta (B)
\end{equation}
be the set of player $2$'s mixed actions that best reply against some (mixed) actions of player $1$'s.
Since player $2$s are myopic, their actions
belong to $\mathcal{B}^*$ at every on-path history.
Player $1$'s minmax payoff is:
\begin{equation}\label{mixmin}
 \underline{v}_1^{min} \equiv  \min_{\beta \in  \mathcal{B}^*} \max_{a \in A} u_1(a, \beta).
\end{equation}
One can show that
$\underline{v}_1^{min}$
is player $1$'s
lowest payoff in any BNE of the reputation game.

I introduce the definition of \textit{monotone-supermodular games}, in which player $1$'s minmax payoff coincides with his lowest stage-game Nash Equilibrium payoff. This class of games have been a primary focus of the reputation literature and are applicable to the study of business transactions (Mailath and Samuelson 2001, Ekmekci 2011, Liu 2011, Liu and Skrzypacz 2014), capital taxation (Phelan 2006), monetary policy (Barro and Gordon 1983), and so on.
\begin{Condition}\label{con2}
Payoffs are monotone-supermodular if
there exist complete orders on $A$ and $B$:\footnote{This monotone-supermodular condition differs from the one in Pei (2020) that considers interdependent value environments. The requirement that $u_1(a,b)$ having strictly decreasing differences in $(a,b)$ is not required for the results in this section, but is needed for the results on stochastic sampling in Section \ref{sec4}.}
\begin{itemize}
  \item[1.] $u_1(a,b)$ is strictly decreasing in $a$ and is strictly increasing in $b$.
  \item[2.] $u_1(a,b)$ has non-increasing differences, and $u_2(a,b)$ has strictly increasing differences in $(a,b)$.
  \item[3.] $a^*$ is not the lowest element in $A$.
\end{itemize}
\end{Condition}
As an example,
the product choice game in the introduction satisfies monotone-supermodularity
when player $1$'s actions are ranked according to $H \succ L$, and player $2$'s actions are ranked according to $T \succ N$.
I provide economic interpretations of this condition in context of the buyer-seller application.
Let player $1$ be a seller with $a \in A$ interpreted as his effort or the quality of his product. Let each player $2$ be a buyer with $b \in B$ interpreted as the quantity she buys. Monotone-supermodularity requires that (1) it is costly for the seller to exert high effort (or equivalently, supply high quality), but he strictly benefits from buyers' purchases; (2) buyers have stronger incentives to purchase larger quantities when the seller's effort is higher (or equivalently, the seller's product quality is higher); (3) the seller receives more benefit from undercutting quality when a buyer purchases a larger quantity;
(4) exerting the lowest effort is not the seller's optimal commitment action.

Under the first and second requirements in Condition \ref{con2}, player $1$'s minmax payoff coincides with his lowest pure strategy Nash equilibrium payoff in the stage game. Under the third requirement, his minmax payoff is strictly lower than his pure Stackelberg payoff.
\begin{Lemma}\label{L3.1}
If players' payoffs are monotone-supermodular, then $\underline{v}_1=\underline{v}_1^{min}<u_1(a^*,b^*)$.
\end{Lemma}

\paragraph{Short-Run Players' Welfare:} My next result demonstrates the failure of reputation effects from
player $2$s' perspective.
Let  $(a'',b'') \in A \times B$ be the worst pure-strategy stage-game Nash equilibrium for player $2$.
Recall that $\delta_s$ is a planner's discount factor when evaluating player $2$s' welfare. Let
\begin{equation*}
\underline{\delta}' \equiv \left\{ \begin{array}{ll}
\max \Big\{
\frac{\max_{a \in A} u_1(a,b^*)-u_1(a^*,b^*)}{\max_{a \in A} u_1(a,b^*)-u_1(a'',b'')} , \frac{u_1(a'',b'')-u_1(a^*,b'')}{u_1(a^*,b^*)-u_1(a^*,b'')}
\Big\}  & \textrm{ if } u_1(a'',b'') < u_1(a^*,b^*) \\
0 & \textrm{ if } u_1(a'',b'') = u_1(a^*,b^*)
\end{array} \right.
\end{equation*}
I show that player $2$s' welfare can be low regardless of the planner's discount factor:
\begin{Lemma}\label{L3.2}
Under Assumptions \ref{Ass1} and \ref{Ass2}. For every $K \in \mathbb{N}$, $\delta_s \in (0,1)$, and $\varepsilon>0$, there exists $\overline{\pi}_0 \in (0,1)$ such that
for every
$\pi_0 \in (0,\overline{\pi}_0)$ and $\delta \geq \underline{\delta}'$,
there exists $(\sigma_1^{\delta},\sigma_2^{\delta}) \in \textrm{SE}(\delta,\pi_0,K)$, such that:
\begin{equation}\label{3.7}
     \mathbb{E}^{(\sigma_1^{\delta},\sigma_2^{\delta},\pi_0)} \Big[   \sum_{t=0}^{\infty} (1-\delta_s) \delta_s^t u_2(a_t,b_t) \Big] \leq u_2(a'',b'')+\varepsilon.
\end{equation}
\end{Lemma}
The proof is in Appendix \ref{subA.1}. Lemma \ref{L3.2} together with
Theorem \ref{Theorem1} has powerful implications on games with monotone-supermodular payoffs (Condition \ref{con2}). Let $\underline{a} \equiv \min A$ and $\{\underline{b} \} \equiv \textrm{BR}_2(\underline{a})$.
 Condition \ref{con2}
 implies that $a^* \succ \underline{a}$,
and
$(\underline{a},\underline{b})$ is a stage-game Nash equilibrium. As a result, games with monotone-supermodular payoffs automatically satisfy Assumption \ref{Ass2}.
\begin{Theorem}\label{Theorem2}
Suppose players' payoffs are monotone-supermodular and satisfy Assumption \ref{Ass1}.
For every $K \in \mathbb{N}$, $\delta_s \in (0,1)$, and $\varepsilon>0$, there exists $\overline{\pi}_0 \in (0,1)$ such that for every
$\pi_0 \in (0,\overline{\pi}_0)$ and $\delta \geq \max\{ \underline{\delta},\underline{\delta}'\}$,
there exists $(\sigma_1^{\delta},\sigma_2^{\delta}) \in \textrm{SE}(\delta,\pi_0,K)$ such that:
\begin{enumerate}
  \item player $1$'s discounted average payoff is $u_1(\underline{a},\underline{b})$,
  \item player $2$'s discounted average welfare is less than $u_2(\underline{a},\underline{b})+\varepsilon$.
\end{enumerate}
\end{Theorem}
Theorem \ref{Theorem2} implies
that slow observational learning leads to sequential equilibria in which \textit{both players receive
low payoffs}.
Under an additional requirement that $\underline{a}$ minimizes $u_2(a,b)$ for every $b \in B$ (i.e., the seller's lowest effort minimizes the buyer's payoff), both players' equilibrium payoffs are arbitrarily close to their respective minmax payoffs.
\subsection{Common Properties of All Equilibria}\label{sub3.3}
To appreciate the subtlety of Theorems \ref{Theorem1} and \ref{Theorem2}, and to distinguish their underlying mechanism from the existing results on social learning and reputation failures,
I establish three properties that apply to \textit{all} Bayes Nash equilibria. I present these findings in decreasing level of generality.

First, I show that player $2$s never herd on actions other than $b^*$ when they have not ruled out the commitment type. It clarifies the conceptual difference between my reputation failure result and the results in
Banerjee (1992), Bikhchandani, Hirshleifer and Welch (1992), and Smith and S{\o}rensen (2000), in which inefficiencies are caused by myopic players herding on an action that mismatches the state.
Formally, for every $\sigma_2\in \Sigma_2$, $h^t \in \mathcal{H}^t$, and $b \in B$,  I say that player $2$s \textit{herd} on action $b$ at $h^t$ if $\sigma_2(h^s)=b$ for every \textit{on-path history} $h^s \succeq h^t$.
\begin{Proposition}\label{Prop1}
Under Assumptions 1 and 2. For every Bayes Nash equilibrium $(\sigma_1,\sigma_2)$, every on-path history $h^t$, and every $b \neq b^*$.
If $\pi(h^t)>0$, then player $2$s cannot herd on $b$ at $h^t$.
\end{Proposition}
\begin{proof}[Proof of Proposition \ref{Prop1}:]
Suppose toward a contradiction that player $2$s herd on $b \neq b^*$ at $h^t$, then
the strategic type has no intertemporal incentive at $h^t$
and at every  $h_*^t$ that differs from $h^t$ only in $\{a_{0},...,a_{t-K}\}$.
In equilibrium, strategic-type player $1$ plays his myopic best reply to $b$ at those histories.  Consider two cases.
First, suppose $\textrm{BR}_1(b)= \{a^*\}$, then in equilibrium,
both types of player $1$ play $a^*$ at $h^t$ and at every $h_*^t$ that differs from $h^t$ only in $\{a_{0},...,a_{t-K}\}$.
As a result, player $2_{t}$  has a strict incentive to play $b^*$ instead of $b$ at $h^t$. This contradicts the presumption that $b \neq b^*$.
Second, suppose $\textrm{BR}_1(b) \neq \{a^*\}$, then in equilibrium,
the strategic type has no incentive to play $a^*$ at $h^t$ and at every $h_*^t$ that differs from $h^t$ only in $\{a_{0},...,a_{t-K}\}$.
Since $\pi(h^t)>0$, player $2_{t+1}$'s belief attaches probability $1$ to the commitment type if she observes $a_t=a^*$, and player $1$'s actions from period $t-K+1$ to $t-1$ and player $2$'s actions from period $0$ to $t-1$ are given according to $h^t$. Therefore,
player $2_{t+1}$ plays $b^*$ following the aforementioned observation, which contradicts the presumption that they herd on $b \neq b^*$.
\end{proof}
Next, I focus on games in which $a^*$ does not best reply against any action in $\mathcal{B}^*$. This includes, but not limited to games with monotone-supermodular payoffs.
\begin{Proposition}\label{Prop2}
If $a^*$ does not best reply against any action in $\mathcal{B}^*$,
then for every Bayes Nash equilibrium and at every on-path history $h^t$ with $\pi(h^t)>0$,
\begin{itemize}
  \item \textit{either} $(b_{t+1},...,b_{t+K})$ is informative about player $1$'s action at $h^t$,\\
\textit{or} $\sigma_2 (h^{t+\tau})=b^*$ for every $h^{t+\tau} \succ h^t$ such that
  $1 \leq \tau\leq K$ and $a^*$ is played from $t$ to $t+\tau-1$.
\end{itemize}
\end{Proposition}
Proposition \ref{Prop2} implies that as long as player $1$
imitates the commitment type, \textit{either}
he is guaranteed to receive his Stackelberg payoff $u_1(a^*,b^*)$ in the next $K$ periods,
\textit{or} player $2$s' actions in the next $K$ periods are informative signals about player $1$'s action in the current period.\footnote{Compared to Proposition \ref{Prop1}, Proposition \ref{Prop2} rules out a larger class of strategies, which includes but not limited to player $2$s herding on a pure action. For example, Proposition \ref{Prop2}  rules out situations in which player $2$s play mixed actions, or their strategies depend nontrivially on calendar time or previous player $2$s' actions, but not on their observations of player $1$'s past actions. As a consequence, Proposition \ref{Prop2} also imposes an additional requirement on players' stage-game payoffs, namely, the Stackelberg action is strictly suboptimal for player $1$ in the stage game.}
Player $1$ receives a high payoff in the first case, and information about his current-period behavior is communicated to all future player $2$s in the second case. This is reminiscent of the logic behind Fudenberg and Levine (1989, 1992)'s reputation results,
that by imitating the commitment type, the long-run player either receives a high stage-game payoff, or
can generate a public signal that is informative about his type.
It also distinguishes Theorem \ref{Theorem1}
from existing results on bad reputations
(e.g., Ely and V\"{a}lim\"{a}ki 2003), in which the short-run players can take an action (i.e., refuse to participate) that stops future short-run players from learning while giving the long-run player his minmax payoff.
\begin{proof}[Proof of Proposition 2:]
Let $(\sigma_1,\sigma_2)$ be a Bayes Nash equilibrium. Let $h^t$ be an on-path history with $\pi(h^t)>0$. Suppose $(b_{t+1},...,b_{t+K})$ is \textit{uninformative} about player $1$'s action at $h^t$, then strategic-type player $1$ plays his myopic best reply against player $2$'s action at $h^t$.

Since player $2$'s action belongs to $\mathcal{B}^*$ at every on-path history, and $a^*$ does not best reply against any action in $\mathcal{B}^*$, strategic-type player $1$ has a strict incentive not to play $a^*$ at $h^t$. Given that $\pi(h^t)>0$,
player $2$'s posterior belief attaches probability $1$ to type $\omega^c$
after observing $a^*$ at $h^t$.
For every $\tau \in \{1,2,...,K\}$, player $2$ plays $b^*$ in period $t+\tau$ when she observes $a^*$ in the past $K$ periods.
\end{proof}
In games with monotone-supermodular payoffs,
Proposition \ref{Prop2} implies
a tight lower bound on player $1$'s \textit{asymptotic payoff} when he plays $a^*$ in every period. Let $\mathbb{E}^{(a^*,\sigma_2)}$ be the expectation when
player $1$ plays $a^*$ in every period and player $2$s use strategy $\sigma_2$. The result is stated as Proposition \ref{Prop3}:
\begin{Proposition}\label{Prop3}
If stage-game payoffs are monotone-supermodular,
then for every BNE $(\sigma_1,\sigma_2)$,
\begin{equation}\label{as}
\liminf_{t \rightarrow \infty}\frac{1}{t} \mathbb{E}^{(a^*,\sigma_2)} \Big[
    \sum_{s=0}^{t-1} u_1(a_s,b_s)
    \Big] \geq \frac{K}{K+1} u_1(a^*,b^*)+\frac{1}{K+1} \min_{b \in \mathcal{B}^*} u_1(a^*,b).\footnote{Proposition \ref{Prop3} does not contradict the disappearing reputation result in Cripps, Mailath and Samuelson (2004) that the long-run player's asymptotic payoff equals to one of his payoffs in the repeated complete information game without commitment type. Notice that the LHS of (\ref{as}) is player $1$'s asymptotic payoff when he plays $a^*$ in every period, while Cripps, Mailath and Samuelson (2004)'s result examines player $1$'s asymptotic payoff under his \textit{equilibrium strategy}. }
 \end{equation}
When $\pi_0$ is small enough and $\delta$ is large enough, there exists a sequential equilibrium with strategy profile $(\sigma_1,\sigma_2)$ such that:
\begin{equation}\label{ass}
\lim_{t \rightarrow \infty}\frac{1}{t} \mathbb{E}^{(a^*,\sigma_2)} \Big[
    \sum_{s=0}^{t-1} u_1(a_s,b_s)
    \Big] =\frac{K}{K+1} u_1(a^*,b^*)+\frac{1}{K+1} \min_{b \in \mathcal{B}^*} u_1(a^*,b).
 \end{equation}
\end{Proposition}
The proof is in Appendix \ref{subA.2}. Inequality (\ref{as}) implies that by imitating the commitment type, player $1$'s asymptotic payoff is at least a fraction $\frac{K}{K+1}$ of his Stackelberg payoff, with the RHS of (\ref{as}) converging to $u_1(a^*,b^*)$ as $K$ goes to infinity.
This lower bound is tight in the sense that when the prior probability of commitment type is below some cutoff and player $1$'s discount factor is above some cutoff,
the RHS of (\ref{ass}) is attained in an equilibrium where the stage-game equilibrium outcome $(\underline{a},\underline{b})$ is played when calendar time is divisible by $K+1$, and the Stackelberg outcome $(a^*,b^*)$ is played when calendar time is not divisible by $K+1$.
Intuitively, when player $1$ plays $a^*$ in a period that is divisible by $K+1$, player $2$s in the next $K$ periods are convinced that player $1$ is committed and will play $b^*$ in response (as long as they do not observe actions other than $a^*$).

Proposition \ref{Prop3} contrasts to Theorem \ref{Theorem1} which suggests that
no matter how large $K$ is,
player $1$'s \textit{discounted average payoff} equals his lowest stage-game equilibrium payoff.
This comparison not only highlights the importance of player $1$'s discount factor in models where the speed of learning is endogenous, but also
shows that Theorem \ref{Theorem1} is driven by \textit{slow convergence} to the high-payoff outcome, rather than low-payoff outcomes in the long run.

\subsection{Proof of Theorem 1}\label{sub3.4}
Recall the definitions of $(a',b')$ and $(a^*,b^*)$.
If $b'=b^*$, then according to Assumption \ref{Ass1}, $a'=a^*$ and payoff $\underline{v}_1$ is attained in an equilibrium where $(a^*,b^*)$ is played at every history.

In what follows, I consider the interesting case in which $b' \neq b^*$.  Assumption \ref{Ass1} implies that:
\begin{equation}\label{3.2}
    u_1(a^*,b^*)\underbrace{>}_{\textrm{P1's pure Stackelberg action is unique}} u_1(a',b')\underbrace{>}_{\textrm{$(a',b')$ is a strict equilibrium}} u_1(a^*,b').
\end{equation}
Let $q^* \in (0,1)$ be such that $b'$ is player $2$'s best reply against player $1$'s mixed action $q^* \circ a^* + (1-q^*) \circ a'$.
Recall that $\overline{\pi}_0\in (0,1)$ is the upper bound on the prior probability of commitment type, which is chosen according to:
\begin{equation*}
    \frac{\overline{\pi}_0}{1-\overline{\pi}_0} = \Big(\frac{q^*}{2-q^*}\Big)^{K+1}.
\end{equation*}
For every $\pi_0 \in (0,\overline{\pi}_0)$ and $\delta$ large enough, I construct the following \textit{three-phase equilibrium} in which strategic-type player $1$ attains payoff $\underline{v}_1$. The current phase of play depends only on the history of player $2$s' actions, which is commonly observed by both players. In period $t$,
\begin{enumerate}
  \item play is in the  \textit{reputation building phase}  if there exists no $s<t$ such that $b_s=b^*$;
  \item play is in the \textit{reputation maintenance phase} if (1) there exists $s<t$ such that $b_s=b^*$, and (2)
  $b_{\tau}=b^*$ for all
  $\tau \in \{s^*+1,...,t-1\}$, where $s^*$ is the smallest $s \in \mathbb{N}$ such that $b_s=b^*$;
  \item play is in the \textit{punishment phase} if (1) there exists $s<t$ such that $b_s=b^*$, and (2) there exists $\tau \in \{s^*+1,...,t-1\}$ such that $b_{\tau} \neq b^*$, where $s^*$ is the smallest $s \in \mathbb{N}$ such that $b_s=b^*$.
\end{enumerate}
Play starts from the reputation building phase, and gradually reaches the reputation maintenance phase. Play reaches the punishment phase only at off-path histories.\footnote{The punishment phase only occurs off the equilibrium path when player 1 knows player 2's sample and players have access to a public randomization device. Punishment occurs \textit{on the equilibrium path} either when player 2's sample is stochastic and is not observed by player 1 (Section \ref{sec4}), or players do not have access to a public randomization device.}

\paragraph{Equilibrium Strategies:} Let $r \in (0,1)$ be defined via the following equation:
\begin{equation}\label{3.13}
    (1-\delta) u_1(a^*,b')+\delta r u_1(a^*,b^*)
    +\delta (1-r) u_1(a',b')
    =u_1(a',b').
\end{equation}
One can verify that when $\delta$ is large enough, $r$ is strictly between $0$ and $1$, and converges to $0$ as $\delta$ goes to $1$.
At every history $h^t$ of the \textit{reputation-building phase},
\begin{itemize}
\item If (1) $t=0$, or (2) $t \geq 1$ and $a_{t-1} \neq a^*$, or (3) $t \geq 1$, $a_{t-1}=a^*$ and $\xi_t>r$, then
player $2$ plays $b'$ and strategic-type player $1$ plays $a^*$ with probability $q^*/2$ and plays $a'$ with probability $1-q^*/2$.
\item If $t  \geq 1$, $a_{t-1}=a^*$ and $\xi_t \leq r$, then player $1$ plays $a^*$ and player $2$ plays $b^*$.
\end{itemize}
At every history of the \textit{reputation maintenance phase},
\begin{itemize}
  \item If $a_{t-1}=a^*$, then player $1$ plays $a^*$ and player $2$ plays $b^*$.
  \item If $a_{t-1}\neq a^*$, then strategic-type player $1$ plays $a'$ and player $2$ plays $b'$.
\end{itemize}
At every history of the \textit{punishment phase}, player $1$ plays $a'$ and player $2$ plays $b'$.

\paragraph{Incentive Constraints:} I verify players' incentives constraints. To start with, when $\delta>\underline{\delta}$,
\begin{equation*}
u_1(a^*,b^*) \geq    (1-\delta) \max_{a \in A} u_1(a,b^*) +\delta u_1(a',b').
\end{equation*}
This implies that player $1$ has an incentive to play $a^*$ in the reputation-maintenance phase. Next, given that player $1$'s continuation value is $u_1(a^*,b^*)$ in the reputation-maintenance phase and is $u_1(a',b')$ in the reputation-building phase,
(\ref{3.13}) implies that player $1$ is indifferent between $a'$ and $a^*$ in the reputation-building phase. Since $a'$
is a strict best reply against $b'$ in the stage game, player $1$
strictly prefers $a'$ to actions other than $a^*$ and $a'$ in the repeated game.

Next, I verify player $2$'s incentive to play $b'$ in the reputation-building phase by
showing that player $1$'s reputation is at most $q^*/2$.
If player $1$ has played actions other than $a^*$ in at least one of the last $K$ periods, then player $2$'s belief attaches probability $0$ to the commitment type.
Suppose $b^*$ has never been played before (i.e., play remains in the reputation-building phase) and $a^*$ was played in the last $K$ periods, then
 $\pi_t$ satisfies the following equation derived from Bayes Rule:
\begin{equation}\label{3.14}
\frac{\pi_t}{1-\pi_t} \Big/ \frac{\pi_0}{1-\pi_0}
    = \frac{\Pr^{(a^*,\sigma_2^{\delta})}(a^*,...,a^*)}{\Pr^{(\sigma_1^{\delta},\sigma_2^{\delta})}(a^*,...,a^*)} \cdot
    \frac{\Pr^{(a_1^*,\sigma_2^{\delta})}(b',...,b',\xi_t | a^*,...,a^*)}{\Pr^{(\sigma_1^{\delta},\sigma_2^{\delta})}(b',...,b', \xi_t | a^*,...,a^*)},
\end{equation}
where $\Pr^{(\sigma_1^{\delta},\sigma_2^{\delta})} (\cdot)$ is the probability measure induced by
strategy profile $(\sigma_1^{\delta},\sigma_2^{\delta})$,
and $\Pr^{(a^*,\sigma_2^{\delta})} (\cdot)$ is the probability measure when player $1$ plays $a^*$ in every period and player $2$s' strategy is $\sigma_2^{\delta}$.

Since the strategic type plays $a^*$ with probability $q^*/2$ in every period of the reputation building phase, we have:
\begin{equation}\label{3.15}
    \frac{\Pr^{(\sigma_1^{\delta},\sigma_2^{\delta})}(a^*,...,a^*|\omega^c)}{\Pr^{(\sigma_1^{\delta},\sigma_2^{\delta})}(a^*,...,a^*|\omega^s)} \leq \Big(\frac{q^*}{2-q^*}\Big)^{-K}.
\end{equation}
In addition, $b'$ occurs with weakly lower probability
when player $1$ is the commitment type, which implies that:
\begin{equation}\label{3.16}
    \frac{\Pr^{(\sigma_1^{\delta},\sigma_2^{\delta})}(b',...,b', \xi_t | a^*,...,a^*,\omega^c)}{\Pr^{(\sigma_1^{\delta},\sigma_2^{\delta})}(b',...,b', \xi_t | a^*,...,a^*,\omega^s)} \leq 1.
\end{equation}
Since $\frac{\pi_0}{1-\pi_0} \leq \frac{\overline{\pi}_0}{1-\overline{\pi}_0}=\Big(\frac{q^*}{2-q^*}\Big)^{K+1}$,
(\ref{3.14}), (\ref{3.15}) and (\ref{3.16}) together imply that
$\pi_t \leq \frac{q^*}{2}$. As a result, the unconditional probability with which player $2$ believes that player $1$ plays $a^*$ is at most $\frac{q^*}{2} +(1-\frac{q^*}{2}) \frac{q^*}{2} < q^*$ at every history of the reputation-building phase, which
verifies player $2$'s incentive to play $b'$.

\section{Reputation Failure under Stochastic Sampling}\label{sec4}
In many applications of interest, each consumer communicates with her neighbors in a social network, learns about the seller's actions against them before making a decision.\footnote{A potential concern is that in practice, when consumer $i$ communicates with consumer $j$, consumer $j$ not only tells consumer $i$ about her personal experience, but also tells consumer $i$ what she learnt from other consumers. If this is the case, then consumers may observe the seller's actions in unboundedly many periods despite each consumer only communicates with a bounded number of predecessors.
This concern is alleviated once consumers' communication costs and capacity constraints are taken into account. For example, (1) due to the cost of communication, a sender is willing to tell a receiver at most a bounded number of the seller's past actions, (2) due to consumers' capacity constraints to process detailed information, a receiver stops sampling other consumers after she has learnt $K$ actions of the seller's in previous periods. Both constraints are relevant for the leading application on retail markets.}
Alternatively, consumers \textit{stochastically} sample their predecessors to learn about their personal experiences.
Compared to the baseline model,
a common feature in these scenarios
is that the seller
\textit{cannot} observe who does each buyer sample or
the realized social network among buyers.
As a result, buyers are \textit{privately learning} about the seller's type and are \textit{privately monitoring} the seller's past actions.

Motivated by these practical concerns,
this section studies games with monotone-supermodular stage-game payoffs,\footnote{In Appendix \ref{subB.1}, I further generalize Theorem \ref{Theorem3} by relaxing the monotone-supermodularity assumption on payoffs. I identify the assumptions that are needed on players' payoffs. The general version of my results apply not only to monotone-supermodular games (e.g., product choice game, capital taxation game, and trust game), but also to common interest games, coordination games, chicken games, and so on.}
and generalizes the insights of Theorem \ref{Theorem1}
when each player $2$ samples a \textit{bounded stochastic subset} of her predecessors and observes player $1$'s actions against them.

Formally, let
$\mathcal{N} \equiv \{\mathcal{N}_t\}_{t=1}^{\infty}$ be a stochastic network among player $2$s, with $\mathcal{N}_t \in \Delta \big(2^{\{0,1,...,t-1\}}\big)$.
The realization of $\mathcal{N}_t$ is denoted by $N_t \subset \{0,1,...,t-1\}$, which
is privately observed by player $2$ who arrives in period $t$ and is unbeknownst to player $1$ and other player $2$s.
Player $2$'s \textit{private history} consists of $N_t$, her predecessors' actions from period $0$ to $t-1$, and player $1$'s actions in periods belonging to $N_t$, i.e.,
\begin{equation}\label{4.3}
 h_2^t \equiv \Big\{N_t,
 b_0,b_1,...,b_{t-1}, \Big(a_s\Big)_{s \in N_t}, \xi_t
 \Big\}.
\end{equation}
Player $1$'s private history
is the same as in the baseline model, i.e., $h_1^t \equiv \{ a_0,...,a_{t-1},b_0,...,b_{t-1},\xi_0,...,\xi_t\}$.
I introduce a regularity condition on the stochastic network $\mathcal{N}$:
\begin{Assumption}\label{Ass4}
For every $s \neq t$, $\mathcal{N}_s$ and $\mathcal{N}_t$ are independent random variables.
There exist $K \in \mathbb{N}$ and $\gamma \in (0,1)$ such that $\Pr \big( |\mathcal{N}_t| \leq K \big) =1$ and $\Pr \big(
  t-1 \in \mathcal{N}_t
  \big) \geq \gamma$ for every $t \geq 1$.
\end{Assumption}
The first part of Assumption \ref{Ass4}
requires different player $2$s'
neighborhoods to be independent. My result extends when each player $2$'s sampling process depends on her predecessors' actions, i.e.,
$\mathcal{N}_t$ depends on $\{b_0,...,b_{t-1}\}$, as long as for every $s<t$, $\mathcal{N}_s$ and
$\mathcal{N}_t$ are independent conditional on $\{b_0,...,b_{s-1}\}$.
This independence assumption is standard in the observational learning literature, which is trivially satisfied when the network is deterministic (Banerjee 1992), and is also assumed in models with random networks such as
Banerjee and Fudenberg (2004) and
Acemoglu, Dahleh, Lobel and Ozdaglar (2011).
It implies that each player $2$ can learn about player $1$'s type only through player $1$'s actions against her neighbors $\{a_s\}_{s \in N_t}$
and her predecessors' actions $\{b_s\}_{s \leq t-1}$, but \textit{cannot}
obtain additional information from the realization of $\mathcal{N}_t$.

The second part $\Pr \big( |\mathcal{N}_t| \leq K \big) =1$ requires the existence of a \textit{uniform upper bound} on the number of predecessors each player $2$ can sample. Similar to the baseline model, $K$ is interpreted as
a constraint on buyers' capacity to process detailed information about the seller's action, or a constraint induced by the costs of communication.
The third part
$\Pr \big(
  t-1 \in \mathcal{N}_t
  \big) \geq \gamma$ requires  a \textit{uniform lower bound} on the probability with which
each player $2$ observes player $1$'s action against her
 immediate predecessor.


Let $\textrm{SE}(\delta,\pi_0,\mathcal{N})$ be the set of sequential equilibrium strategy profiles in a repeated game with social network $\mathcal{N}$, discount factor $\delta$, and prior belief $\pi_0$. Theorem \ref{Theorem3} extends Theorem \ref{Theorem1} to reputation games with stochastic network monitoring:
\begin{Theorem}\label{Theorem3}
If stage-game payoffs are monotone-supermodular and satisfy Assumption \ref{Ass1}, and the stochastic network $\mathcal{N}$ satisfies Assumption \ref{Ass4},
then there exist $\overline{\pi}_0 \in (0,1)$ and $\underline{\delta} \in (0,1)$,
such that for every $\pi_0 \in (0,\overline{\pi}_0)$ and $\delta > \underline{\delta}$,
there exists $(\sigma_1^{\delta},\sigma_2^{\delta}) \in \textrm{SE}(\delta,\pi_0,\mathcal{N})$, such that:
\begin{equation}\label{4.5}
     \mathbb{E}_1^{(\sigma_1^{\delta},\sigma_2^{\delta})} \Big[   \sum_{t=0}^{\infty} (1-\delta) \delta^t u_1(a_t,b_t) \Big] =\underline{v}_1.
\end{equation}
\end{Theorem}
The proof is in Appendix B. Several economically interesting features in the construction of Section \ref{sub3.4} extend.
First, play starts from an \textit{active learning phase} in which player $1$ receives a low payoff, and gradually enters a \textit{reputation maintenance phase} in which learning stops and player $1$ receives a high payoff. Second, the transition probability between the two phases depends on player $1$'s discount factor, which vanishes to zero as $\delta \rightarrow 1$. Third,
the asymptotic play converges to the Stackelberg outcome $(a^*,b^*)$ with probability close to $1$. However, the speed with which play converges to this high-payoff phase is low. This wipes out a patient player $1$'s gains from building reputations and
discourages him from playing the Stackelberg action and investing in his reputation.

A major technical challenge stems from \textit{private monitoring} and \textit{private learning}, both of which arise when the seller does not observe the realized network. The belief-free approach to construct equilibria in Ely, H\"{o}rner and Olszewski (2005),
H\"{o}rner and Lovo (2009), and H\"{o}rner, Lovo and Tomala (2011)
does not apply, since player $2$s are myopic in my model and have no intertemporal incentive. In equilibria where nontrivial learning takes place, player $2$s' actions are sensitive to their posterior beliefs about player $1$'s type, making it hard to sustain belief-free incentives.

To illustrate the idea of my proof, consider the product choice game in the introduction. Let $q^*$ be
the minimal probability with which $H$ needs to be played in order to provide player $2$ an incentive to play $T$.
Let $\{a_0,...,a_{t-1},b_0,...,b_{t-1}\}$ be a \textit{complete history} in period $t$.

First, consider providing belief-free incentives such that (1) conditional on each \textit{complete history}, player $2$ believes that $H$ will be played with probability $q^*$, and (2) each player $2$ mixes between $T$ and $N$ with probability that makes player $1$ indifferent between $H$ and $L$. Under this arrangement,
both $T$ and $N$ are player $2$'s best replies, regardless of her belief about player $1$'s type and private history.

However, this belief-free arrangement is feasible in period $t$ only if after observing
\textit{each} complete history in period $t$,
player $2$'s posterior belief attaches probability less than $q^*$ to the commitment type.
Since myopic player $2$s play $N$ in the active learning phase, the probability with which player $1$ plays $H$ is at most $q^*$. Therefore,
a hypothetical observer's posterior belief attaches probability arbitrarily close to $1$ to the commitment type after observing a sufficiently long string of $H$.
This implies the existence of a cutoff calendar time, such that the above belief-free arrangement is feasible only if calendar time is below this cutoff.

In light of this observation, I use the following \textit{belief-based construction} when calendar time is above the aforementioned cutoff.
In particular,
player $1$'s action depends on his private history, which
is chosen such that each player $2$ is indifferent under her posterior belief about player $1$'s private history.
This is equivalent to establish the existence of solution to a system of linear equations, in which the number of player $1$'s private histories is the number of free variables, and the number of player $2$'s private histories is the number of linear constraints.
In period $t$, the number of free variables is $2^t$. Given that
each player $2$'s
sample size is at most $K$, the number of constraints is at most
$2^K \sum_{j=0}^K {t \choose j}$.
An important observation is that the linear system is under-determined \textit{if and only if} $t$ is large relative to $K$.
This explains why the belief-free construction is used only when calendar time is below the cutoff, and the belief-based construction is used only when calendar time is above the cutoff.

\section{Private Signals about Current Period Action}\label{sec5}
This section studies a variant of the baseline model in which every player $2$ observes \textit{an informative private signal} about player $1$'s action \textit{in the current period}, in addition to the entire history of her predecessors' actions, and possibly, player $1$'s actions in the last $K \in \{0\} \cup \mathbb{N}$ periods (or player $1$'s actions against her neighbors in some stochastic network that satisfies Assumption \ref{Ass4}).

I provide necessary and sufficient conditions under which the patient player can secure his Stackelberg payoff in all equilibria (Theorems \ref{Theorem4} and 4').
I also discuss policy implications based on the comparison between the results in this section and the ones in Sections \ref{sec3} and \ref{sec4}.

Players move \textit{sequentially} in the stage game. In period $t$, player $1$ chooses $a_t \in A$ after observing $h_1^t$. Before choosing $b_t \in B$, player $2$ observes a noisy private signal about $a_t$ denoted by $s_t \in S$, the entire history of player $2$s' past actions, and potentially, player $1$'s actions in the last $K$ periods. Formally, let
    $h_2^t \equiv \{s_t,b_0,...,b_{t-1},a_{t-K},...,a_{t-1}\}$.
    Player $1$'s private history remains the same as in the baseline model.
    Players' strategies and strategy sets are defined accordingly.

Let $f(\cdot|a_t)$ be the distribution of $s_t$, with
$\mathbf{f} \equiv \{f(\cdot|a)\}_{a \in A}$.
Let $\textrm{NE}(\delta,\pi_0,K, \mathbf{f})$ be the set of BNEs, and
\begin{equation}
    \underline{V}_1(\delta, \pi_0,K,\mathbf{f}) \equiv  \inf_{(\sigma_1,\sigma_2) \in \textrm{NE}(\delta,\pi_0,K,\mathbf{f})}
    \mathbb{E}_1^{(\sigma_1,\sigma_2)} \Big[   \sum_{t=0}^{\infty} (1-\delta) \delta^t u_1(a_t,b_t) \Big].
\end{equation}
be player $1$'s worst BNE payoff. Let $\textrm{SE}(\delta,\pi_0,K, \mathbf{f})$ be the set of sequential equilibrium strategy profiles.

\paragraph{Games with $|A|=2$:}
I start from games in which player $1$'s action choice is binary (for example, the product choice game in the introduction), while $|B|$ and $|S|$ can be arbitrary. I introduce the notions of boundedly informative signals and unboundedly informative signals, and characterize the set of $\mathbf{f}$ such that a patient player $1$ can secure his Stackelberg payoff in all BNEs.
\begin{Condition}[Boundedly Informative \textit{vs} Unboundedly Informative Signals]
For given $a \in A$, I say that $\mathbf{f}$ is \textit{unboundedly informative} about $a$ if there exists $s \in S$ such that $f(s|\widetilde{a})>0$ if and only if $\widetilde{a}=a$. Otherwise, $\mathbf{f}$ is \textit{boundedly informative} about $a$.
\end{Condition}
\begin{Theorem}\label{Theorem4}
Suppose $|A|=2$.
\begin{itemize}
  \item[1.] If $\mathbf{f}$ is unboundedly informative about $a^*$, then for every $K \in \mathbb{N} \cup\{0\}$ and $\pi_0 >0$,
  \begin{equation}
 \liminf_{\delta \rightarrow 1}   \underline{V}_1(\delta, \pi_0,K,\mathbf{f})  \geq u_1(a^*,b^*).
\end{equation}
  \item[2.] If $\mathbf{f}$ is boundedly informative about $a^*$ and players' stage-game payoffs satisfy Assumptions \ref{Ass1} and \ref{Ass2}, then for every $K \in \mathbb{N} \cup\{0\}$, there exists $\overline{\pi}_0 \in (0,1)$ such that for every $\pi_0 \in (0,\overline{\pi}_0)$
and $\delta$ large enough,
there exists $(\sigma_1^{\delta},\sigma_2^{\delta}) \in \textrm{SE}(\delta,\pi_0,K,\mathbf{f})$, such that:
\begin{equation}
     \mathbb{E}_1^{(\sigma_1^{\delta},\sigma_2^{\delta})} \Big[   \sum_{t=0}^{\infty} (1-\delta) \delta^t u_1(a_t,b_t) \Big] = \underline{v}_1.
\end{equation}
\end{itemize}
\end{Theorem}
Theorem \ref{Theorem4} suggests that in games where $|A|=2$ and player $1$'s Stackelberg payoff is strictly greater than $\underline{v}_1$,
player $1$ can secure his Stackelberg payoff \textit{if and only if} player $2$'s private signal about his action is unboundedly informative.
The sufficiency part of this result applies to \textit{all} monitoring structures about player $1$'s past actions. For example,  player $1$ can secure his Stackelberg payoff even when player $2$s cannot observe any of his past actions (i.e., $K=0$), observe noisy signals about his past actions, or observe his past actions according to any arbitrary stochastic network $\mathcal{N}$.

Intuitively, player $2$ observing an unboundedly informative private signal about $a_t$ guarantees a \textit{lower bound on the speed of observational learning} that \textit{does not} depend on player $1$'s discount factor.
This accelerates the process of reputation building, and makes it worthwhile for player $1$ to establish a good reputation. When $s_t$ is boundedly informative about $a^*$, the minimal speed of learning vanishes to zero as $\delta$ goes to $1$.
Similar to Theorem \ref{Theorem1},
the prolonged learning process wipes out player $1$'s gains from building reputations, no matter how patient he is.

The requirement of unboundedly informative signals in Theorem \ref{Theorem4} is reminiscent of a well-known result in Smith and S{\o}rensen (2000),\footnote{My result
does not follow from Mirrlees (1976) who shows that in principal-agent models, the principal can implement the first best outcome when there exists a signal realization that occurs with zero probability when the agent takes the first-best action and occurs with positive probability otherwise.
This is because in my model, the rewards and punishments to player $1$ are dictated by future player $2$s' behaviors. Depending on the equilibrium being played, there are \textit{multiple ways} in which the signal realizations can be mapped into player $1$'s continuation payoffs. When $\delta$ is close to $1$, it is unclear whether player $1$ has an incentive to play $a^*$ and attain his Stackelberg payoff $u_1(a^*,b^*)$ in \textit{all equilibria}.} that
in canonical social learning models, myopic players'
actions are asymptotically efficient
\textit{if and only if} their private signals are unboundedly informative about the payoff-relevant state.
Theorem \ref{Theorem4} is conceptually different from their result for two reasons.

First, as suggested by Theorem \ref{Theorem1} and Proposition \ref{Prop3}, converging to a high-payoff outcome asymptotically is \textit{insufficient} for player $1$ to receive a high discounted average payoff, no matter how close his discount factor is to $1$. This is demonstrated by the constructed equilibria in the proof of Theorem \ref{Theorem1}, in which
a patient player's discounted average payoff is low despite his asymptotic payoff is high.

Second, the short-run players in my model learn about the \textit{endogenous actions} of a strategic long-run player, while in Smith and S{\o}rensen (2000), they learn about  an \textit{exogenous state}.
Even when $\mathbf{f}$ is unboundedly informative about $a^*$, there is no guarantee that $b_t$ is informative about $a_t$, or $b_t$ is informative about player $1$'s type in periods where player $1$ receives a low stage-game payoff. An example is provided later in this section, in which
$\mathbf{f}$ is unboundedly informative about $a^*$, but
$b_t$ is uninformative about player $1$'s type despite $b^*$ is played with probability bounded away from $1$.

In addition, a technical difficulty arises when $K \geq 1$. Player $2$s in my model can entertain \textit{heterogenous beliefs} about the informativeness of previous  player 2s' actions, while in
Smith and S{\o}rensen (2000), the informativeness of each public signal is common knowledge.
Intuitively, each player 2 privately observes  player 1's actions in the last $K$ periods, which are not observed by other player 2s.
Due to the potential serial correlation in player 1's equilibrium actions,
the informativeness of $b_t$ about $a_t$ can be different under the private beliefs of different player 2s.

My proof (Appendix \ref{secC}) addresses these concerns by making use of
 three observations.
First, regardless of player $2$'s prior belief about player $1$'s action, she
has a \textit{strict incentive} to play $b^*$
after observing the signal realization that  occurs only when $a=a^*$.
This highlights the role of unboundedly informative signals, which contrasts to
boundedly informative signals
and the baseline model in which player $2$s do not receive any private signal about player $1$'s current-period action.

Second, when player $1$'s action choice is binary and $\mathbf{f}$ is unboundedly informative about $a^*$, the following likelihood ratio that measures the informativeness of player $2$'s action
    \begin{equation*}
            \frac{\Pr(b_t=b^*|a_t=a^*)}{\Pr(b_t=b^*|a_t\neq a^*)}
    \end{equation*}
is bounded away from $1$ whenever the ex ante probability with which $b_t=b^*$ is bounded away from $1$.
As a result, the informativeness of $b_t$ about $a_t$ is bounded from below according to the private belief of player $2$ who arrives in period $t$.

Third, if the ex ante probability that $b_t=b^*$ is bounded away from $1$ and $b_t$ is informative about player $1$'s type according to period $t$ player $2$'s private belief,
then the informativeness of $b_t$ about $\omega$ is also uniformly bounded from below according to the private beliefs of \textit{all} future player $2$s.

For an intuitive explanation,  suppose player $2$ who arrives in period $t$ observes that $a^*$ has been played in the last $K$ periods,
and before observing the realization of $s_t$,
 believes that she will play $b^*$ with probability bounded away from $1$. Then the probability with which $(a_{t-K},...,a_{t-1})=(a^*,...,a^*)$ under strategic-type player $1$'s equilibrium strategy is bounded away from zero. Otherwise,
 after
 period $t$ player $2$ observes $(a_{t-K},...,a_{t-1})=(a^*,...,a^*)$, her posterior belief attaches probability close to $1$ to the commitment type, and the probability that she plays $b^*$ in period $t$ cannot be bounded away from $1$.
Therefore, the probability with which player $2$ in period $s \neq t$ believes that $(a_{t-K},...,a_{t-1})=(a^*,...,a^*)$ occurring with very low probability is uniformly bounded from above.

The above reasoning suggests that in all equilibria,
if strategic-type player $1$ deviates from his equilibrium strategy and imitates the commitment type,
then in every period $t \in \mathbb{N}$, \textit{either} player $2$ plays $b^*$ with probability close to $1$,
\textit{or} every player $2$ who arrives after period $t$ attaches probability bounded from below to period $t$ player $2$'s true private history.
As a result, as long as the informativeness of $b_t$ about $a_t$ is bounded away from $0$ according to period $t$ player $2$'s private belief,
it is uniformly bounded away from $0$ according to \textit{all} future player $2$s' private beliefs.


\paragraph{Games with $|A| \geq 3$:} In games where player $1$ has three or more actions, the equivalence between guaranteeing
high returns from  building reputations
and unboundedly informative private signal about the Stackelberg action breaks down. To illustrate, consider the following $2 \times 3$ game:
\begin{center}
\begin{tabular}{| c | c | c |}
\hline
  - & $b^*$ & $b'$  \\
  \hline
  $\overline{a}$ & $1,4$ & $-2,0$ \\
  \hline
  $a^*$ & $2,1$ & $-1,0$  \\
  \hline
    $\underline{a}$ & $3,-2$ & $0,0$ \\
  \hline
\end{tabular}
\end{center}
Let $S \equiv \{\overline{s},s^*,\underline{s}\}$, with $f(s^*|a^*)=2/3$, $f(\underline{s}|a^*)=1/3$, $f(\overline{s}|\overline{a})=1$, $f(\overline{s}|\underline{a})=1/3$, and $f(\underline{s}|\underline{a})=2/3$.

One can verify that players' stage-game payoffs
satisfy Assumptions \ref{Ass1} and \ref{Ass2}, and
are monotone-supermodular when the order on player $1$'s actions is
$\overline{a} \succ a^* \succ \underline{a}$, and
the order on player $2$s' actions is  $b^* \succ b'$.
Player $1$'s  Stackelberg action is $a^*$, his Stackelberg payoff is $2$, and
$\mathbf{f}$ is unboundedly informative about $a^*$.

I construct a sequential equilibrium in which player $1$'s payoff is bounded away from $2$.  Strategic-type player $1$ plays a mixed action that depends only on player $2$'s posterior belief about his type. If player $2$'s posterior belief assigns probability $\pi$ to the commitment type, then the strategic-type player $1$ plays $\alpha (\pi) \in \Delta (A)$, which is pinned down by:
\begin{equation*}
   (1-\pi) \circ \alpha(\pi) +\pi \circ a^* = 0.5 \circ a^* + 0.25 \circ \overline{a}+ 0.25 \circ \underline{a}.
\end{equation*}
Player $2$ plays $b^*$ if $s_t \in \{s^*,\overline{s}\}$ and plays $b'$ if $s_t=\underline{s}$.

Player $1$'s payoff under this strategy profile is $1$, which is strictly below his pure Stackelberg payoff $2$.
This strategy profile is an equilibrium since player $1$'s expected stage-game payoff
is $1$ no matter which action he plays, and his continuation payoff is independent of his action in the current period. Player $2$ has a strict incentive to play $b^*$ after observing $\overline{s}$ or $s^*$,
and has an incentive to play $b'$ after observing  $\underline{s}$.
Conditional on each type of player $1$, the probability
with which player $2$ plays $b^*$ is $2/3$, i.e., $b_t$ is uninformative about player $1$'s type.

In this example, $b_t$ is uninformative about player $1$'s type despite $\mathbf{f}$ is unboundedly informative about $a^*$
and the ex ante probability with which $b_t=b^*$ is bounded away from $1$.
This is driven by the  \textit{heterogeneity} in player $2$'s incentive to play $b^*$ against different actions of player $1$'s.
In particular, player $2$ has a stronger incentive to play $b^*$ when player $1$ plays $\overline{a}$ than when player $1$ plays $a^*$.
As a result, player $2$ has an incentive to play $b^*$ following a signal realization $\overline{s}$ that occurs with lower probability under $a^*$, and has an incentive to play $b'$ following a signal realization $\underline{s}$ that occurs with higher probability under $a^*$. This concern never arises when $|A|=2$ given there is only one action in $A$ that is not the Stackelberg action. However, it can happen in games where $|A| \geq 3$.

Motivated by the buyer-seller applications, I focus on games with \textit{monotone-supermodular payoffs} and extend Theorem \ref{Theorem4} to games in which player $1$ can have any number of actions, and
the signal distribution $\mathbf{f}$
satisfies a standard \textit{monotone likelihood ratio property} (or MLRP).
\begin{Condition}[MLRP]
Given a complete order on $A$,
$\mathbf{f}$ has MLRP if there exists a complete order on $S$, such that:
\begin{equation}\label{4.22}
    \frac{f(s|a)}{f(s'|a)} \geq  \frac{f(s|a')}{f(s'|a')} \textrm{ for every }  a \succ a' \textrm{ and } s \succ s'.
\end{equation}
\end{Condition}
In applications to retail markets where a patient seller chooses the quality he supplies and each buyer along a sequence chooses whether to trust the seller after observing an informative private signal about the seller's action in the current period, MLRP requires the buyers' private signal realizations to be ranked such that a higher signal realization indicates that the product is of higher quality.

When player $1$'s actions are ranked according to $\overline{a} \succ a^* \succ \underline{a}$,\footnote{Under alternative complete orders on $A$, players' payoffs are not monotone-supermodular.}
the previous $2 \times 3$ example violates MLRP regardless of the complete order on $S$.
This is because $\overline{s}$ occurs with strictly positive probability under $\overline{a}$ and $\underline{a}$, but occurs with zero probability under $a^*$. As a result, there exist bad equilibria in which player $2$'s actions are uninformative about player $1$'s action,  patient player $1$ receives a low payoff from building reputations, and has weak incentives to invest in his reputation.
My characterization result for games with $|A| \geq 3$ is stated as Theorem 4':
\begin{Theorem4}
Suppose players' payoffs are monotone-supermodular and
$\mathbf{f}$ satisfies MLRP.
\begin{itemize}
  \item[1.] If $\mathbf{f}$ is unboundedly informative about $a^*$, then for every $K \in \mathbb{N} \cup\{0\}$ and $\pi_0 >0$,\\
   \begin{equation}
   \liminf_{\delta \rightarrow 1}   \underline{V}_1(\delta,\pi_0,K,\mathbf{f})  \geq u_1(a^*,b^*).
   \end{equation}
  \item[2.] If $f(\cdot|a)$ has full support for every $a \in A$  and players' payoffs satisfy Assumption \ref{Ass1}, then for every $K \in \mathbb{N} \cup\{0\}$, there exists $\overline{\pi}_0 \in (0,1)$ such that for every $\pi_0 \in (0,\overline{\pi}_0)$
and $\delta$ large enough,
there exists $(\sigma_1^{\delta},\sigma_2^{\delta}) \in \textrm{SE}(\delta,\pi_0,K,\mathbf{f})$, such that:
\begin{equation}
     \mathbb{E}_1^{(\sigma_1^{\delta},\sigma_2^{\delta})} \Big[   \sum_{t=0}^{\infty} (1-\delta) \delta^t u_1(a_t,b_t) \Big] = \underline{v}_1.
\end{equation}
\end{itemize}
\end{Theorem4}
Theorem 4' suggests that in games with monotone-supermodular payoffs and the signal distribution satisfies MLRP, $\mathbf{f}$ being unboundedly informative about $a^*$ is sufficient and almost necessary for player $1$ to secure his Stackelberg payoff in all equilibria of the reputation game. The $2\times 3$ game example earlier in this section demonstrates why the MLRP requirement is indispensable. In Appendix \ref{subC.4}, I use an example to explain why the full support condition in statement 2 of Theorem 4' cannot be replaced by bounded informativeness when player $1$ has three or more actions.

The proof is similar to that of Theorem \ref{Theorem4}, with the differences explained in Appendix \ref{subC.3}.
For statement 1, the key is to show that for every prior belief about player $1$'s action $\alpha \in \Delta (A)$ with $a^* \in \textrm{supp}(\alpha)$, and every best reply
$\beta: S \rightarrow \Delta (B)$
of player $2$'s against $\alpha$, if the ex ante probability with which $b_t=b^*$ is bounded away from $1$, then the relative entropy between the distribution over $b$ induced by $(\alpha,\beta)$ and that induced by $(a^*,\beta)$ is bounded away from $0$.

I explain the role of unbounded informativeness and MLRP in deriving this reputation result.
Since players' payoffs are monotone-supermodular and $\mathbf{f}$ satisfies MLRP, player $2$ has an incentive to play $b^*$ only when the realization of $s$ belongs to some interval $[\underline{s},\overline{s}]$.
Since $\mathbf{f}$ is unboundedly informative, there exists $s^*$ such that $f(s^*|a)>0$ if and only if $a=a^*$, i.e.,
the probability that $s \neq s^*$ is strictly lower when player $1$ plays $a^*$ compared to any other action.
MLRP also implies that the likelihood ratio between $s \in (s^*,\overline{s}]$ and $s>\overline{s}$ is decreasing in $a$,
and the likelihood ratio between $s \in [\underline{s},s^*)$ and $s<\underline{s}$ is increasing in $a$.
As a result, for any $\alpha' \in \Delta (A\backslash\{a^*\})$ and $\beta: S \rightarrow \Delta (B)$
that is non-decreasing
in $s$ and is not constantly $b^*$ in the support of $f(\cdot|\alpha')$, the probability with which player $2$ plays $b^*$ is strictly higher
when player $1$ plays $a^*$ relative to $\alpha'$.
This implies that player $2$'s action is informative about player $1$'s type as long as
her ex ante probability of playing $b^*$ is bounded away from $1$.

\section{Policy Implications}\label{sec6}
Theorems \ref{Theorem4} and 4'  shed light on the effectiveness of various policies in accelerating social learning and encouraging sellers to supply high quality products.

To illustrate, consider a regulator who faces a budget constraint in every period and can inspect at most a fraction $\epsilon \in (0,1)$ of products currently sold on the market. In another word, a seller's product quality in any given period is known to the market with probability at most $\epsilon$.

Theorem \ref{Theorem4} suggests that such a policy is effective in accelerating learning when the regulator issues quality certificates to the inspected products that have high quality. Given the presence of observational learning (i.e., consumers observing each other's actions), this policy restores patient sellers' incentives to build reputations even when only the consumers who demand products in the current period can notice this certificate.
By contrast, informing consumers only about the low-quality products that are screened out
cannot effectively accelerate learning, and the market may remain in a bad equilibrium with slow learning and weak incentives to establish good reputations.

Broadly speaking,
my model and policy implications fit into retail markets in developing countries, in which there is significant asymmetric information, lack of formal records, and persistent mistrust between buyers and sellers. As substitutes for official records, consumers acquire information from their peers, such as observing others' choices and learning about others' experiences through word-of-mouth communication.

An example in which my modeling assumptions fit is the watermelon retail market studied in Bai (2018). My results suggest an explanation to the findings in her field experiment. To start with, in the baseline setting prior to policy interventions, buyers' choices respond slowly to their acquaintances' past experiences,
and sellers exert little effort on sorting to improve the quality of their melons.\footnote{The watermelon retail market is highly localized, with most consumers coming from the same neighborhood and knowing each other. The seller can improve the quality of his melons by exerting effort on sorting when procuring melons from the wholesale market.
Bai (2018)'s structural estimation results rule out several alternative explanations, such as reputation fails due to the seller's impatience, noisy monitoring of the seller's actions, and so on.} My Theorems \ref{Theorem1} and \ref{Theorem3} suggest a rationale for her observations of slow learning, persistent mistrust, and weak incentives to build good reputations.

She then conducts a randomized control trial that provides different sellers with different branding technologies.
She finds that among the group of sellers who are
provided with novel laser-cut labels, most of them exert high effort on sorting and
trust is gradually built between these sellers and their buyers.
Among the group of sellers who are provided with standard
sticker labels that can be counterfeited, the outcomes are similar to the baseline setting in which sellers are reluctant to build reputations and consumers' skepticism about these sellers' product quality persists over time.

Theorems \ref{Theorem4} and 4' suggest an explanation to the different outcomes under these treatments:
The laser-cut labels cannot be easily forged, and consumers believe that they can only be used on high-quality products. This intervention increases buyers' responsiveness to the seller's actions, guarantees a lower bound on the speed of observational learning regardless of the buyers' prior belief about the seller's action.
The sticker labels are ineffective since they can be counterfeited. When
there is widespread mistrust between buyers and sellers (i.e., buyers' prior belief attaches low enough probability to the commitment type), and buyers entertain the adverse belief that the strategic-type seller is likely to exert low effort, they have a rationale for not trusting the seller
no matter which quality label they observe.
The low responsiveness of their actions
leads to
slow observational learning, which results in low returns from building reputations.

\section{Concluding Remarks}\label{sec7}
This paper examines a long-run player's incentive to build reputations when his opponents have limited observations of his past actions, and learn primarily from other myopic players' actions. One can also view it as a social learning model in which the object to learn is the endogenous action of a strategic long-run player instead of an exogenous state.

My results relate the \textit{speed of learning} and the patient player's \textit{guaranteed returns from a good reputation} to the myopic players' private
signals. In particular, observing an unboundedly informative signal about the patient player's current-period action guarantees a lower bound on the speed of learning and a high return from  building reputations, while observing his actions in a bounded number of previous periods can result in slow learning and low returns from building reputations.

I conclude by  applying
Gossner (2011)'s arguments to my model, which
reconcile the findings in Theorem \ref{Theorem1}, Theorem \ref{Theorem4}, and the reputation results in Fudenberg and Levine (1989, 1992) and Gossner (2011).
I explain why observational learning leads to an uninformative payoff lower bound despite
player $2$s' actions are informative about player $1$'s past actions.
I also explain the conceptual differences between the low-payoff equilibria in my model with the ones in
reputation models with two equally patient players, such as
Cripps and Thomas (1997) and Chan (2000).

\subsection{Relative Entropy \& Value of Reputations}\label{sub7.1}
Recall that Gossner (2011) establishes the following upper bound on the expected sum of Kullback-Leibler divergence (hereafter, KL divergence) between the distribution over public signals generated by the commitment type and the distribution over public signals generated by players' equilibrium strategies:
  \begin{equation}\label{3.8}
    \mathbb{E}^{(a^*,\sigma_2)}\Big[
    \sum_{t=0}^{\infty}
    d\Big(
    y_t(\cdot|a^*)
    \Big|\Big|
    y_t(\cdot)
    \Big)
      \Big]
    \leq -\log \pi_0,
  \end{equation}
where $\pi_0$ is the prior probability of the commitment type,
$y_t(\cdot)$ is  the distribution over period $t$ public signals according to players' equilibrium strategies, $y_t(\cdot|a^*)$ is the distribution over period $t$ public signals when player $2$ plays her equilibrium strategy and player $1$ plays $a^*$ in every period,
and $d(\cdot\|\cdot)$ is the relative entropy between the two probability distributions. I call
$d\big(
    y_t(\cdot|a^*)
    \big|\big|
    y_t(\cdot)
    \big)$
player $2$'s \textit{one-step-ahead prediction error} in period $t$.

Inequality (\ref{3.8}) applies to my setting once we take $y_t$ to be the distribution of
$(b_{t+1},...,b_{t+K})$. According to Proposition \ref{Prop2}, $(b_{t+1},...,b_{t+K})$ is informative about $a_t$ unless
player $1$'s average payoff in the next $K$ periods is at least
$u_1(a^*,b^*)$.

The difference arises when deriving the \textit{lower bound on player $1$'s discounted average payoff} from inequality (\ref{3.8}).
In the models of Fudenberg and Levine (1989, 1992) and Gossner (2011), if the public signals can statistically identify player $1$'s actions and when player $2$ does not have a strict incentive to play $b^*$,
then
$d\big( y_t(\cdot|a^*)\big\| y_t(\cdot) \big)$ is bounded from below by a strictly positive number.
Therefore, as long as player $1$ imitates the commitment type,
the expected number of periods in which player $2$s' myopic best reply is not $b^*$
is bounded from above. As player $1$ becomes patient, the payoff consequence of this bounded number of periods goes
to $0$. As a result, a patient player $1$ is guaranteed to receive his optimal commitment payoff.

In my model where players' stage-game payoffs are monotone-supermodular, the value of $d\big( y_t(\cdot|a^*)\big\| y_t(\cdot) \big)$ is \textit{strictly positive}
whenever player $2$ does not have a strict incentive to play $b^*$, and player $1$'s average payoff from period $t$ to $t+K$ is less than
    $\frac{K}{K+1} u_1(a^*,b^*)+\frac{1}{K+1} \min_{b \in \mathcal{B}^*} u_1(a^*,b)$.

However, the lower bound on $d\big( y_t(\cdot|a^*)\big\| y_t(\cdot) \big)$ depends on $\delta$, and vanishes to $0$ as
$\delta \rightarrow 1$.
Intuitively, future player $2$s' actions are responsive to player $1$'s past actions in order to provide player $1$ an incentive to play actions other than his myopic best reply (for example, player $1$'s Stackelberg action).
When player $1$ becomes more patient, he puts more weight on his continuation value relative to his stage-game payoff. Therefore, he is willing to sacrifice his stage-game payoff even when his action affects player $2$s' future actions with low probability.
This endogenously reduces the informativeness of player $2$s' actions, lowers the speed of learning,
increases the amount of time required for player $1$ to establish a reputation, which in turn, wipes out player $1$'s returns from building reputations.

In the constructed equilibrium in Section \ref{sub3.3},
player $2$'s one-step-ahead prediction error in the reputation building phase is:
\begin{equation}\label{3.9}
    \log \Big(1+ (1-q^*) (1-\delta)\Big).
\end{equation}
Using the Taylor's expansion, (\ref{3.9})
is of magnitude $1-\delta$ when $\delta$ is close to $1$.
As a result, when player $1$ imitates the commitment type, the expected number of periods with which player $2$'s belief about player $1$'s action being far away from $a^*$ goes to infinity as $\delta \rightarrow 1$. As predicted by Theorem \ref{Theorem1},
the negative payoff consequence of such periods offsets the benefits from building reputations.

In the proofs of Theorem \ref{Theorem4} and 4', the key step is to show that
player $2$'s one-step-ahead prediction error is bounded away from zero in all periods where $b_t=b^*$
occurs with probability bounded away from $1$. Since the lower bound on the one-step-ahead prediction error is uniformly bounded away from zero for all values of $\delta$, the expected number of periods where player $2$ does not play $b^*$ is bounded. This suggests that a patient player $1$ can guarantee his Stackelberg payoff by building a reputation.

\subsection{Reputation Models with Two Equally Patient Players}\label{sub7.2}
Cripps and Thomas (1997) study reputation games
between an informed player and an \textit{equally patient} uninformed player.
They focus on \textit{common interest games}, and assume that
both players can perfectly observe each other's actions in the past.
When the prior probability of commitment type is sufficiently low, they construct a sequential equilibrium in which both players' payoffs are arbitrarily close to their minmax payoffs. In the following example, suppose with small but positive probability, player $1$ is a commitment type that plays $H$ in every period,
\begin{center}
\begin{tabular}{| c | c | c |}
  \hline
  -- & $A$ & $B$ \\
  \hline
  $H$ & $1,1$ & $-\varepsilon,-\varepsilon$ \\
  \hline
  $L$ & $-\varepsilon,-\varepsilon$ & $0,0$ \\
  \hline
\end{tabular}
\end{center}
there \textit{exist} equilibria in which
both players' discounted average payoffs are arbitrarily close to $0$.

In the active learning phase of their constructed equilibrium, player $1$ plays $H$ with probability close to $1$, i.e., learning is slow. However, player $2$ does not play her myopic best reply against $H$.
The reason is: player $2$ fears that
once she plays $A$ while player $1$ reveals rationality (i.e., by playing $L$),
the inefficient outcome $(L,B)$ will be played in all future periods.
Similar to the equilibrium constructed in the proof of Theorem \ref{Theorem1},
the asymptotic play also converges to the Stackelberg outcome although the patient player's discounted average payoff is low.
The above finding is generalized by Chan (2000) to all games except for
(1) games where player $1$ has a strictly dominant action, and (2) games with strictly conflicting interests.

Compared to the reputation failure results of Cripps and Thomas (1997) and Chan (2000) that hinge on  the uninformed player's patience, I show that learning can be arbitrarily slow and reputation effects can fail even when the uninformed players are myopic.

In terms of how players' patience affects the informed player's guaranteed payoff from building reputations,
a general lesson from models with \textit{unbounded records} is that
the informed player's patience helps reputation building, while the uninformed player's patience undermines reputation building. Indeed, when player $2$ becomes more patient, the conditions under which player $1$ can secure his commitment payoff becomes more stringent. For example, the commitment payoff theorem requires no condition on stage-game payoffs in long-run short-run models (e.g., Fudenberg and Levine 1989, 1992), requires the game to have conflicting interest in long-run medium-run models  (e.g., Schmidt 1993), and requires the game to have strictly conflicting interest in long-run long-run models (e.g., Chan 2000, Cripps, Dekel and Pesendorfer 2005).
This contrasts to my model in which the \textit{informed player's patience} is self-defeating since it endogenously lowers the speed of social learning. This causes reputations to fail
even when the uninformed players are myopic.

In addition,
my reputation failure result also applies to games in which Chan (2000)'s folk theorem result fails. This includes games with
strictly conflicting interests, such as:
\begin{center}
\begin{tabular}{| c | c | c |}
  \hline
  -- & $Out$ & $In$ \\
  \hline
  $F$ & $2,0$ & $0,-1$ \\
  \hline
  $A$ & $2,0$ & $1,1$ \\
  \hline
\end{tabular}
\end{center}
When there exists a commitment type that  plays
$F$ in every period,
Cripps, Dekel and Pesendorfer (2005)
show that player $1$ can secure his Stackelberg payoff $2$ in a model with equally patient players.

By contrast, my Theorem \ref{Theorem1} suggests that when player $2$s are myopic, have unlimited observation of other player $2$s' actions, but have bounded observations about player $1$'s past actions,
there exist equilibria in which  player $1$'s payoff equals his minmax payoff $1$.\footnote{The example in Cripps, Dekel and Pesendorfer (2005) violates Assumption 1 since both $F$ and $A$ are player $1$'s stage-game best replies against player $2$'s action $Out$. Nevertheless, one can use the same argument in the proof of Theorem \ref{Theorem1} to find equilibria in which player $1$'s discounted average payoff is $1$. The details are available upon request.}

\newpage
\appendix
\section{Proofs in Section 3}\label{secA}
\subsection{Proof of Lemma 3.2 and Theorem 2}\label{subA.1}
I show Lemma \ref{L3.2}. For Theorem \ref{Theorem2}, notice that $(\underline{a},\underline{b})$ is the unique Nash Equilibrium of the stage game, one can obtain a constructive proof to Theorem \ref{Theorem2} by replacing $(a'',b'')$ with $(\underline{a},\underline{b})$.

First, consider the case in which $u_1(a'',b'')=u_1(a^*,b^*)$. Assumption \ref{Ass1} implies that $(a'',b'')=(a^*,b^*)$, and the discounted average welfare of player $2$ equals $u_2(a'',b'')$ in a pooling equilibrium in which $(a^*,b^*)$ is played at every on-path history.

Next, consider the nontrivial case in which $u_1(a'',b'')<u_1(a^*,b^*)$. Consider a similar equilibrium as the proof of Theorem \ref{Theorem1} except for two differences: first, replace $(a',b')$ with $(a'',b'')$, and second, calibrate the probability with which strategic-type player $1$ playing $a^*$ in the reputation-building phase such that the \textit{unconditional probability} of $a^*$ equals $q^*$ at every reputation-building phase history. This is feasible given that player $1$ knows player $2$'s belief about his type at every on-path history.
Let $V_2$ be player $2$'s discounted average welfare in the reputation-building phase (with discount factor $\delta_s$):
\begin{equation}\label{A.1}
    V_2=(1-\delta_s) \Big\{
    q^* u_2(a^*,b'') +(1-q^*) u_2(a'',b'')
    \Big\}
    +\delta_s \Big\{
    (1-q^*) V_2 +q^* (1-r) V_2 +q^* r u_2(a^*,b^*)
    \Big\},
\end{equation}
where
\begin{equation}\label{A.2}
    r \equiv \frac{1-\delta}{\delta} \frac{u_1(a'',b'')-u_1(a^*,b'')}{u_1(a^*,b^*)-u_1(a'',b'')}
\end{equation}
is the transition probability between phases that makes strategic-type player $1$ indifferent between playing $a^*$ and $a''$ in the reputation-building phase. Equation (\ref{A.1}) yields:
\begin{equation}\label{A.3}
    V_2 \Big\{
    1-\delta_s (1-q^*) -\delta_s q^* (1-r)
    \Big\}
    =\delta_s q^* r u_2(a^*,b^*)
    +(1-\delta_s) \Big(
    q^* u_2(a^*,b'')+(1-q^*) u_2(a'',b'')
    \Big).
\end{equation}
Since $q^*$ can be arbitrarily low, when $q^*$ converges to $0$, (\ref{A.3}) reduces to $V_2(1-\delta_s)=u_2(a'',b'') (1-\delta_s)$, which implies that $V_2$ is arbitrarily close to $u_2(a'',b'')$ as $q^*$ and $\overline{\pi}_0$ become arbitrarily small.

\subsection{Proof of Proposition 3}\label{subA.2}
\paragraph{Lower Bound:} Consider the strategic-type's payoff when he deviates and imitates the commitment type.
For every $\beta \in \Delta (B)$ and $a \prec a^*$, MSM implies that $u_1(a^*,\beta)< u_1(a,\beta)$. Let $h^t \equiv \{a_{s},b_s\}_{s=0}^{t-1}$.
For every $t \in \mathbb{N}$ and $a \in A$, let $E_t(a,b^t)$ be the event that (1) player $1$ plays $a$ in period $t$, (2) player $1$ has played $a^*$ from period $t-K+1$ to $t-1$, (3) player $1$ plays according to $\sigma_1$ starting from period $t+1$, and (4) the history of player $2$'s actions until period $t$ is $b^t \equiv (b_0,...,b_{t-1})$.
For every $\tau \in \{1,2,...,K\}$ and $h^t \equiv (a^*,...,a^*,b^t)$,
let $y_t^{\tau}(\cdot|a,h^t) \in \Delta (B)$ be the distribution of $b_{t+\tau}$
conditional on event $E_t(a,b^t)$, and
let $y_t (\cdot|a,h^t) \in \Delta (B^K)$ be the distribution of
$(b_{t+1},...,b_{t+K})$ conditional on event $E_t(a,b^t)$.
Let $\overline{u}_1$ and $\underline{u}_1$ be player $1$'s highest and lowest feasible stage-game payoffs, respectively, and let
$||\cdot||$ be the total variation norm. If
\begin{equation}\label{A.4}
    ||y_t(\cdot|a^*,h^t)-y_t(\cdot|a,h^t)|| \leq \frac{1-\delta}{2\delta (\overline{u}_1-\underline{u}_1)} \Big(
    u_1(a,\beta)- u_1(a^*,\beta)
    \Big),
\end{equation}
then the strategic-type player $1$ has a strict incentive to play $a$ instead of $a^*$ at $h^t$ as well as at every history $h_*^t$ that differs from $h^t$ only in terms of $\{a_0,...,a_{t-K}\}$. The latter is because the distribution of $\{b_{t+1},...,b_{t+K}\}$ does not depend on
$\{a_0,...,a_{t-K}\}$ since they cannot be observed by players $2_{t+1}$ to $2_{t+K}$.
Let
\begin{equation}\label{A.5}
    \Delta \equiv \frac{1-\delta}{2K\delta (\overline{u}_1-\underline{u}_1)} \min_{\beta \in \Delta (B), a \prec a^*} \Big\{ u_1(a,\beta)- u_1(a^*,\beta) \Big\}.
\end{equation}
Since
\begin{equation*}
    ||y_t^{\tau}(\cdot|a^*,h^t)-y_t^{\tau}(\cdot|a,h^t)||  \leq ||y_t(\cdot|a^*,h^t)-y_t(\cdot|a,h^t)|| \leq \sum_{s=1}^{K} ||y_t^{s}(\cdot|a^*,h^t)-y_t^{s}(\cdot|a,h^t)||,
\end{equation*}
inequality (\ref{A.4}) holds when
$||y_t^{\tau}(\cdot|a^*,h^t)-y_t^{\tau}(\cdot|a,h^t)|| \leq \Delta$ for every $\tau \in \{1,2,...,K\}$.
Let $\mathcal{H}^{(a^*,\sigma_2)}$ be the set of public histories that occur with positive probability
when player $1$ plays $a^*$ in every period and player $2$ plays $\sigma_2$. I partition $\mathcal{H}^{(a^*,\sigma_2)}$ into two subsets, $\mathcal{H}^{(a^*,\sigma_2)}_0$ and $\mathcal{H}^{(a^*,\sigma_2)}_1$:
\begin{enumerate}
  \item[1.] If there exists $a \prec a^*$ such that $||y_t^{\tau}(\cdot|a^*,h^t)-y_t^{\tau}(\cdot|a',h^t)|| \leq \Delta$ for every $\tau$, then
  $h^t \in \mathcal{H}_0^{(a^*,\sigma_2)}$.
  \item[2.] If for every $a \prec a^*$, there exists $\tau$ such that
   $||y_t^{\tau}(\cdot|a^*,h^t)-y_t^{\tau}(\cdot|a',h^t)|| \geq \Delta$, then $h^t \in \mathcal{H}_1^{(a^*,\sigma_2)}$.
\end{enumerate}
For every $h^t \in \mathcal{H}_0^{(a^*,\sigma_2)}$, the strategic type has a strict incentive not to play $a^*$ at $h^t$, which means that player $2$ attaches probability $1$ to the commitment type after observing $a^*$ at $h^t$. For every $\tau \in \{1,2,...,K\}$, every on-path history $h^{t+\tau} \succ h^t$ such that $a^*$ has been played from period $t$ to $t+\tau-1$, player $2$ has a strict incentive to play $b^*$ at $h^{t+\tau}$.
This in addition to the fact that player $2$ plays an action at least as large as $b'$ at every on-path history implies that for every $h^t \in \mathcal{H}_0^{(a^*,\sigma_2)}$, we have:
\begin{equation}\label{A.6}
    \frac{1}{K+1} \mathbb{E}^{(a^*,\sigma_2)} \Big[
    \sum_{s=t}^{t+K} u_1(a_s,b_s)
   \Big| h^t \Big] \geq \frac{K}{K+1} u_1(a^*,b^*)+\frac{1}{K+1} u_1(a^*,b').
\end{equation}

For every $h^t\in \mathcal{H}_1^{(a^*,\sigma_2)}$, there exists a constant $\gamma >0$ such that for every $\alpha \in \Delta (A)$ such that $b \prec b^*$ best replies against $\alpha$, we have $||y_t(\cdot|a^*,h^t)-y_t(\cdot|\alpha,h^t)|| \geq \gamma \Delta$. The Pinsker's inequality implies that
\begin{equation}\label{A.7}
    d\Big( y_t(\cdot|\alpha,h^t) \Big\| y_t(\cdot|a^*,h^t) \Big) \geq 2\gamma^2 \Delta^2.
\end{equation}
for every such $\alpha \in \Delta (A)$.
For every equilibrium $(\sigma_1,\sigma_2)$ and every $\tau \in \{0,1,...,K\}$,
\begin{equation}\label{A.8}
    \mathbb{E}^{(a^*,\sigma_2)} \Big[
    \sum_{s=0}^{\infty} d \Big(
     y_{s(K+1)+\tau}(\cdot|\sigma_1(h^{s(K+1)+\tau}),h^{s(K+1)+\tau})
    \Big\| y_{s(K+1)+\tau}(\cdot|a^*,h^{s(K+1)+\tau})
    \Big)
    \Big] \leq -\log \pi_0.
\end{equation}
Inequalities (\ref{A.7}) and (\ref{A.8}) together imply that:
\begin{equation}\label{A.9}
    \mathbb{E}^{(a^*,\sigma_2)} \Big[
    \sum_{s=0}^{\infty} \mathbf{1}\Big\{ h^{s(K+1)+\tau} \in \mathcal{H}_1^{(a^*,\sigma_2)} \textrm{ and }
    \sigma_2(h^{s(K+1)+\tau}) \prec b^*
    \Big\}
    \Big] \leq -\frac{\log \pi_0}{2\gamma^2 \Delta^2}
\end{equation}
I derive a lower bound for $\liminf_{t \rightarrow \infty}\frac{1}{t} \mathbb{E}^{(a^*,\sigma_2)} \Big[\sum_{s=0}^{t-1} u_1(a_s,b_s)\Big]$ using inequalities (\ref{A.6}) and (\ref{A.9}).
For every $\tau \in \{0,1,...,K\}$, let
 \begin{equation*}
\mathcal{H}_0^{\tau}   \equiv  \Big\{
    h^t \Big| \exists h^{s(K+1)+\tau} \in \mathcal{H}_0^{(a^*,\sigma_2)} \textrm{ such that } h^t \succeq h^{s(K+1)+\tau} \textrm{ and } t \in [s(K+1),s(K+1)+K]
    \Big\},
\end{equation*}
let
\begin{equation*}
    \mathcal{H}_1^{\tau} \equiv \Big\{
    h^{s(K+1)+\tau} \in \mathcal{H}_1^{(a^*,\sigma_2)} \Big| s \in \mathbb{N}
    \Big\},
 \end{equation*}
and let  $\mathcal{H}^{\tau} \equiv \mathcal{H}_0^{\tau}  \cup \mathcal{H}_1^{\tau}$. By definition,
 $\mathcal{H}^{(a^*,\sigma_2)} = \bigcup_{\tau=0}^{K}    \mathcal{H}^{\tau}$. An important observation is that for every
 $\tau,\tau' \in \{0,1,...,K\}$ with
 $\tau \neq \tau'$,
\begin{equation}\label{A.10}
 \mathcal{H}_1^{\tau} \cap \mathcal{H}_1^{\tau'}=\{\varnothing\} \textrm{ and }   \mathcal{H}_0^{\tau} \cap \mathcal{H}_0^{\tau'}=\{\varnothing\}.
\end{equation}
The former is straightforward.
For the latter, suppose toward a contradiction that $h^t \in \mathcal{H}_0^{\tau} \cap \mathcal{H}_0^{\tau'}$ with $\tau < \tau'$,
there exist $h^s$ and $h^{s+\tau'-\tau}$ such that
$h^{t} \succsim h^{s+\tau'-\tau}  \succ h^s$, $h^{s} \in \mathcal{H}_0^{\tau}$,
$t-s \leq K$,
and $s-\tau$ is divisible by $K+1$. On one hand $h^s \in \mathcal{H}_0^{\tau}$ and
$\tau'-\tau \leq K$
implies that $\sigma_1(h^{s+\tau'-\tau}) = a^*$.
On the other hand $h^{s+1} \in \mathcal{H}_0^{\tau'}$ implies that $\sigma_1(h^{s+\tau'-\tau}) \neq a^*$. This leads to a contradiction.

For every $\tau \in \{0,1,...,K\}$, inequality (\ref{A.6}) implies that player $1$'s expected average payoff at histories in $\mathcal{H}_0^{\tau}$ is at least the RHS of (\ref{as}). Since $\mathcal{H}_0^{\tau} \cap \mathcal{H}_0^{\tau'}=\{\varnothing\}$ for every $\tau \neq \tau'$, it implies that
player $1$'s expected average payoff at histories in $\bigcup_{\tau=0}^{K} \mathcal{H}_0^{\tau}$ is at least the RHS of (\ref{as}). For every $\tau \in \{0,1,...,K\}$,  (\ref{A.9}) implies that player $1$'s expected average payoff at histories belonging to set
    $\mathcal{H}_1^{\tau} \Big\backslash
\bigcup_{s=0}^{K} \mathcal{H}_0^{s}$
is at least $u_1(a^*,b^*)$. Since $\mathcal{H}_1^{\tau} \cap \mathcal{H}_1^{\tau'}=\{\varnothing\}$ for every $\tau \neq \tau'$, it implies that
 player $1$'s expected average payoff at histories in
 $\bigcup_{s=0}^{K} \mathcal{H}_1^{s} \Big\backslash
\bigcup_{s=0}^{K} \mathcal{H}_0^{s}$
is at least $u_1(a^*,b^*)$. The two parts imply that
\begin{equation*}
    \liminf_{t \rightarrow \infty}\frac{1}{t} \mathbb{E}^{(a^*,\sigma_2)} \Big[
    \sum_{s=0}^{t-1} u_1(a_s,b_s)
    \Big] \geq \frac{K}{K+1} u_1(a^*,b^*)+\frac{1}{K+1}  u_1(a^*,b').
\end{equation*}

\paragraph{Tightness of Lower Bound:}  When payoffs are monotone-supermodular, $(a',b')$ is the unique stage-game Nash equilibrium. Let $\overline{\pi}_0$ be the largest real number in $(0,1)$ such that $b'$ best replies against the mixed action $\overline{\pi}_0 \circ a^* +(1-\overline{\pi}_0) \circ a'$. Consider the following construction when $\pi_0 \in (0,\overline{\pi}_0)$.
At every on-path history (the set of on-path histories can be derived recursively),
\begin{itemize}
  \item if $t$ is divisible by $K+1$, then
player $1$
plays $a'$ and player $2$ plays $b'$ in period $t$;
  \item if $t$ is not divisible by $K+1$, then
 player $1$
plays $a^*$ and player $2$ plays $b^*$ in period $t$.
\end{itemize}
I partition off-path histories into three subsets. For every period $t$ public history such that:
\begin{itemize}
  \item (1) there exists no $r<t$, such that $b_r \neq b^*$ and $r$ is not divisible by $K+1$; (2) there exists no $s<t$ such that
$b_s \neq b'$ and $s$ is divisible by $K+1$; (3) player $2$ observes player $1$ playing an off-path action in period $t-1$,
then players play $(a^*,b^*)$
 if $t$ is divisible by $K+1$, and play $(a',b')$ if $t$ is not divisible by $K+1$.
  \item (1) there exists no $r<t$, such that $b_r \neq b^*$ and $r$ is not divisible by $K+1$, but (2) there exists $s<t$ such that $b_s \neq b'$ and $s$ is divisible by $K+1$. If $t-1$ is divisible by $K+1$, $b_{t-1}=b^*$ while $a_{t-1} \neq a^*$, then play $(a',b')$ in period $t$. If $t-1$ is divisible by $K+1$, $b_{t-1}=b^*$ while $a_{t-1}=a^*$, then play $(a^*,b^*)$ in period $t$ if and only if $\xi_t >1/2$ and play $(a',b')$ in period $t$ otherwise. If $t-1$ is not divisible by $K+1$, or $b_{t-1} \neq b^*$, then play $(a^*,b^*)$ if $t$ is not divisible by $K+1$ and play $(a',b')$ if $t$ is divisible by $K+1$.
  \item there exists $r<t$, such that $b_r \neq b^*$ and $r$ is not divisible by $K+1$, then play $(a',b')$ in all subsequent periods.
\end{itemize}
Player $1$'s time-average payoff from playing $a^*$ in every period equals the RHS of (\ref{as}). I verify players' incentive constraints.
Since $b^*$ best replies to $a^*$ and $b'$ best replies to $a'$, player $2$'s incentive constraints are satisfied. In what follows, I verify player $1$'s incentives. At every on-path history $h^t$,
\begin{itemize}
  \item If $t+1$ not divisible by $K+1$ and $t$ is not divisible by $K+1$, then the strategic type's continuation value from playing $a^*$ in period $t$ is at least
      \begin{equation}\label{A.11}
           V \equiv   \frac{u_1(a',b')+\delta u_1(a^*,b^*)+\delta^2 u_1(a^*,b^*)+...+\delta^K u_1(a^*,b^*)}{1+\delta+...+\delta^K},
      \end{equation}
      while his continuation value from playing any other action is $u_1(a',b')$. This verifies his incentive to play $a^*$ when $\delta$ is above some cutoff.
  \item If $t+1$ not divisible by $K+1$ and $t$ is divisible by $K+1$, then the strategic type's continuation values from playing $a^*$ and $a'$ are the same, equal $V$, while his continuation value from playing other actions is $u_1(a',b')$. He has a strict incentive to play $a'$ since $a'$ best replies to $b'$.
  \item If $t+1$ is divisible by $K+1$, then the strategic type's continuation value from playing $a^*$ in period $t$ is at least $V$. If he deviates and plays $a_t$, then consider his incentive in period $t+1$ at off-path history $(h^t,a_t,b_t=b^*)$.

      Since player $2$ plays $b^*$ in period $t+1$ after observing player $1$'s deviation in period $t$, player $1$'s continuation value from playing $a^*$ in period $t+1$ is at least $\frac{1}{2}V+\frac{1}{2} u_1(a',b')$. This is because player $2$ will play $b^*$ with probability $1/2$ in period $t+2$, after which player $1$ will be forgiven for his deviation. Player $1$'s continuation value from playing actions other than $a^*$ in period $t+1$ is $u_1(a',b')$. Therefore, he has a strict incentive to play $a^*$ in period $t+1$ following his deviation in period $t$, and his continuation value in period $t$ when he deviates is strictly lower than $V$.
\end{itemize}

\subsection{Sufficient Conditions for Attaining Minmax Payoff}\label{subA.3}
Focusing on games in which $\underline{v}_1^{min}<\underline{v}_1$, I identify
a sufficient condition under which player $1$'s lowest equilibrium payoff in the reputation game coincides with his minmax payoff. Recall the definition of $\mathcal{B}^*$ in (3.2) and minmax payoff in (3.3).
Action $\beta \in \mathcal{B}^*$ is
player $2$'s \textit{minmax action} if and only if player $1$'s payoff from best replying against $\beta$ equals
$\underline{v}_1^{min}$.
\begin{Condition}
There exists a minmax action $\beta \in \mathcal{B}^*$ such that:
\begin{itemize}
  \item[1.] $b^* \notin \textrm{supp}(\beta)$,
  \item[2.] there exists $\alpha \in \Delta (A)$ with $a^* \in \textrm{supp}(\alpha)$ such that
  $\beta$ best replies against $\alpha$,
  \item[3.] $u_1(\alpha,\beta) \geq u_1(a^*,\beta)$.
\end{itemize}
\end{Condition}
Condition 3 requires that first, there exists a minmax action that excludes the Stackelberg best reply $b^*$
in its support. Second, $\beta$ best replies against a (potentially mixed) action of player $1$'s that includes the Stackelberg action in its support. Third, playing the Stackelberg action against $\beta$ yields player $1$ a weakly lower payoff compared to playing $\alpha$.
The first and third part of this condition is satisfied when the Stackelberg action is costly for player $1$ regardless of player $2$'s action, and player $2$s' Stackelberg best reply is always beneficial to player $1$.
The second part of condition 3 is satisfied for generic $(u_1,u_2)$, since it only requires $\beta$ to be a strict best reply against some $\alpha \in \Delta (A)$.

I provide an example in which
$\underline{v}_1^{min}< \underline{v}_1$ and players' stage-game payoffs
satisfy Condition 3 and Assumptions 1 and 2:
\begin{center}
\begin{tabular}{| c | c | c | c |}
  \hline
  -- & $L$ & $C$ & $R$\\
  \hline
  $U$ & $1,1$ & $0,0$ & $-2,0$ \\
  \hline
  $M$ & $2,0$ & $0,0$ & $-1,1$ \\
  \hline
  $B$ & $0,0$ & $1/2,1/2$ & $0,0$ \\
  \hline
\end{tabular}
\end{center}
Action $R$ belongs to $\mathcal{B}^*$ since $R$ is a strict best reply against $M$, and therefore, best replies against any mixed action in which the probability of $M$ is close to $1$.
Player $1$'s worst stage-game Nash equilibrium payoff is $1/2$, his pure Stackelberg payoff is $1$, and his minmax payoff is $0$.
When player $2$ plays her minmax action $R$, player $1$'s payoff from playing $M$ is strictly greater than his payoff from playing $U$.
The following theorem extends Theorem \ref{Theorem1}, with proof in Appendix \ref{subA.3}:
\begin{Theorem1}
When players' stage-game payoffs satisfy Assumptions \ref{Ass1} and \ref{Ass2} and Condition 3.
For every $K \in \mathbb{N}$, there exists $\overline{\pi}_0 \in (0,1)$,
such that for every $\pi_0 \in (0,\overline{\pi}_0)$ and $\delta \geq \underline{\delta}$,
there exists $(\sigma_1^{\delta},\sigma_2^{\delta}) \in \textrm{SE}(\delta,\pi_0,K)$, such that:
\begin{equation*}
     \mathbb{E}_1^{(\sigma_1^{\delta},\sigma_2^{\delta})} \Big[   \sum_{t=0}^{\infty} (1-\delta) \delta^t u_1(a_t,b_t) \Big] = \underline{v}_1^{min}.
\end{equation*}
\end{Theorem1}
\begin{proof}
When $u_1(a^*,b^*)=\underline{v}_1^{min}$, $(a^*,b^*)$ is the unique pure strategy Nash Equilibrium of the stage game. An equilibrium that attains payoff $\underline{v}_1^{min}$ is that both players play $(a^*,b^*)$ at every history.

Next, I focus on the interesting case in which $u_1(a^*,b^*)>\underline{v}_1^{min}$.
Recall the definitions of $\alpha$ and $\beta$ in Condition 3.
Let $q^*$ be the probability that $\alpha$ attaches to $a^*$. Let $\overline{\pi}_0 \in (0,1)$ be small enough such that:
\begin{equation}
    \frac{\overline{\pi}_0}{1-\overline{\pi}_0} \leq \Big(\frac{q^*}{2-q^*}\Big)^{K+1}.
\end{equation}
For every $\pi_0 \in (0,\overline{\pi}_0)$ and $\delta$ large enough, I construct the following \textit{three-phase equilibrium} in which strategic-type player $1$ attains payoff $\underline{v}_1$. The current phase of play depends only on the history of player $2$s' actions, which are commonly observed by both players. In period $t$,
\begin{enumerate}
  \item play is in the  \textit{reputation building phase}  if there exists no $s<t$ such that $b_s=b^*$;
  \item play is in the \textit{reputation maintenance phase} if (1) there exists $s<t$ such that $b_s=b^*$, and (2) there exists no $\tau \in \{s^*+1,...,t-1\}$ such that $b_{\tau} \neq b^*$, where $s^*$ is the smallest $s$ with $b_s=b^*$;
  \item play is in the \textit{punishment phase} if (1) there exists $s<t$ such that $b_s=b^*$, and (2) there exists $\tau \in \{s^*+1,...,t-1\}$ such that $b_{\tau} \neq b^*$, where $s^*$ is the smallest $s$ such that $b_s=b^*$.
\end{enumerate}
Play starts from the reputation building phase, and gradually reaches the reputation maintenance phase. Play reaches the punishment phase only at off-path histories.

\paragraph{Equilibrium Strategies:} At every history $h^t$ of the \textit{reputation-building phase},
\begin{itemize}
\item If $t=0$, then player $2$ plays $\beta$ and strategic type player $1$ plays $\alpha_0$ that satisfies:
\begin{equation}\label{A.11}
 (1-\pi_0)   \alpha_0 +\pi_0 a^*=\alpha,
\end{equation}
Such $\alpha_0$ exists given that $\pi_0 <q^*/2$, which have been assumed in (A.11).
\item If $t  \geq 1$ and $\xi_t  > r(a_{t-1})$, then player $2$ plays $\beta$ and strategic type player $1$ plays $\alpha(h^t)$ that satisfies:
\begin{equation}\label{A.12}
 (1-\pi(h^t))   \alpha(h^t) +\pi(h^t) a^*=\alpha,
\end{equation}
where $\pi(h^t)$ is the probability player $2$'s belief at $h^t$ attaches to the commitment type.
Such $\alpha(h^t)$ exists given that $\pi(h^t) <q^*/2$, which I will verify by induction by the end of this proof.
\item If $t  \geq 1$ and $\xi_t  \leq r(a_{t-1})$, then player $2$ plays $b^*$ and strategic type player $1$ plays $a^*$.
\end{itemize}
The transition probability to the reputation maintenance phase is a function of player $1$'s action in the previous period $r: A \rightarrow [0,1]$, which is pinned down by the following equation:
\begin{equation}\label{A.13}
    (1-\delta) u_1(a,\beta)+\delta r(a) u_1(a^*,b^*)
    +\delta (1-r(a)) \max_{a \in A} u_1(a ,\beta)
    =\max_{a \in A} u_1(a ,\beta).
\end{equation}
Given that $u_1(a^*,b^*)>\underline{v}_1^{min}$, one can verify that (1) for every $a,a' \in A$, $r(a) \geq r(a')$ if and only if
$u_1(a,\beta) \leq u_1(a',\beta)$, and
(2) for every $a \in A$,
$r(a)$ is strictly between $0$ and $1$
when $\delta$ is large enough, and
as $\delta \rightarrow 1$, the value of $r(a)$ converges to $0$.

At every history of the \textit{reputation maintenance phase}, If $a_{t-1}=a^*$, then player $1$ plays $a^*$ and player $2$ plays $b^*$.
If $a_{t-1}\neq a^*$, then player $2$ plays $\beta$ and strategic-type player $1$ plays $\alpha$, after which the continuation play enters the \textit{punishment phase}. Player $1$'s continuation value in the last period of the reputation maintenance phase equals $\underline{v}_1^{min}$.\footnote{This can be achieved for example, by repeating the strategies in the reputation-building phase.}

\paragraph{Incentive Constraints:} I verify players' incentives constraints. To start with, when $\delta$ is large enough,
\begin{equation*}
u_1(a^*,b^*) \geq    (1-\delta) \max_{a \in A} u_1(a,b^*) +\delta \underline{v}_1^{min}.
\end{equation*}
This implies that player $1$ has an incentive to play $a^*$ in the reputation maintenance phase when $\delta$ is large enough. Next, given that player $1$'s continuation value is $u_1(a^*,b^*)$ in the reputation maintenance phase and is $\underline{v}_1^{min}$ in the reputation building phase,
(\ref{A.13}) implies that
player $1$ is indifferent between all of his actions in the reputation building phase.

Next, I verify player $2$'s incentive to play $\beta$ at the reputation building phase by showing that
(1) player $1$'s reputation at every history of the reputation building phase is less than $q^*/2$, and (2) $\alpha(h^t)$ defined in (\ref{A.12}) attaches probability at least $q^*/2$ to action $a^*$. According to (A.11), player $1$'s reputation in period $0$ is less than $q^*/2$ and according to (\ref{A.11}), $\alpha(h^0)$ attaches probability more than $q^*/2$ to action $a^*$.

Suppose the conclusion holds for all reputation-building phase histories $h^s$ with $s < t$. If $h^t$ belongs to the reputation building phase,
then $b^*$ has never been played before  given that player $2$ plays $\beta$ in every period of the past given that $b^* \notin \textrm{supp}(\beta)$.
Player $2$'s posterior belief about player $1$ being committed is $0$ unless
$a^*$ was played in the last $\min \{K,t\}$ periods, in which case her posterior belief
$\pi_t$ satisfies the following equation:
\begin{equation}\label{A.14}
\frac{\pi_t}{1-\pi_t} \Big/ \frac{\pi_0}{1-\pi_0}
    = \frac{\Pr^{(a^*,\sigma_2^{\delta})}(a^*,...,a^*)}{\Pr^{(\sigma_1^{\delta},\sigma_2^{\delta})}(a^*,...,a^*)} \cdot
    \frac{\Pr^{(a_1^*,\sigma_2^{\delta})}(b_0,...,b_{t-1},\xi_t | a^*,...,a^*)}{\Pr^{(\sigma_1^{\delta},\sigma_2^{\delta})}(b_0,...,b_{t-1}, \xi_t | a^*,...,a^*)},
\end{equation}
where $\Pr^{(\sigma_1^{\delta},\sigma_2^{\delta})} (\cdot)$ is the probability measure induced by
strategy profile $(\sigma_1^{\delta},\sigma_2^{\delta})$,
and $\Pr^{(a^*,\sigma_2^{\delta})} (\cdot)$ is the probability measure when player $1$ plays $a^*$ in every period and player $2$s' strategy is $\sigma_2^{\delta}$.

Since the strategic type plays $a^*$ with probability $q^*/2$ in the reputation building phase, we have:
\begin{equation}\label{A.15}
    \frac{\Pr^{(\sigma_1^{\delta},\sigma_2^{\delta})}(a^*,...,a^*|\omega^c)}{\Pr^{(\sigma_1^{\delta},\sigma_2^{\delta})}(a^*,...,a^*|\omega^s)} \leq \Big(\frac{q^*}{2-q^*}\Big)^{-\min\{t,K\}} \leq \Big(\frac{q^*}{2-q^*}\Big)^K.
\end{equation}
Since $u_1(\alpha,\beta) \geq u_1(a^*,\beta)$,
we have:
\begin{equation*}
    r(a^*) \geq \mathbb{E}[r(\widetilde{a})|\alpha].
\end{equation*}
This implies that for every $(b_0,...,b_{t-1})$ with $b_s \in \textrm{supp}(\beta)$ for every $s$, we have:
\begin{equation}\label{A.16}
    \frac{\Pr^{(\sigma_1^{\delta},\sigma_2^{\delta})}(b_0,...,b_{t-1}, \xi_t | a^*,...,a^*,\omega^c)}{\Pr^{(\sigma_1^{\delta},\sigma_2^{\delta})}(b_0,...,b_{t-1}, \xi_t | a^*,...,a^*,\omega^s)} \leq 1.
\end{equation}
Since $\frac{\pi_0}{1-\pi_0} \leq \frac{\overline{\pi}_0}{1-\overline{\pi}_0}=\Big(\frac{q^*}{2-q^*}\Big)^{K+1}$,
(\ref{A.14}), (\ref{A.15}) and (\ref{A.16}) together imply that
$\pi_t \leq \frac{q^*}{2}$. As a result, the probability with which player $2$ believes player $1$ playing $a^*$ is at most $\frac{q^*}{2} +(1-\frac{q^*}{2}) \frac{q^*}{2} \leq q^*$. This completes the verification of player $2$s' incentives.
\end{proof}

\section{Proof of Theorem 3}\label{secB}
In Appendix \ref{subB.1}, I relax the monotone-supermodularity condition on payoffs, and identify weaker sufficient conditions for my result.
My sufficient conditions are satisfied not only in games with monotone-supermodular payoffs, but also in coordination games, common interest games, and many other games studied in the reputation literature. I state a Theorem 3' that generalizes Theorem \ref{Theorem3} to a larger class of payoff structures.
The proof of Theorem 3' is in Appendices \ref{subB.2} and \ref{subB.3}.

\subsection{Relax Monotone-Supermodularity}\label{subB.1}
Recall that $a^*$ is player $1$'s pure Stackelberg action and $b^*$ is player $2$'s unique best reply against $a^*$. Let
$a''$ be the unique element in $\textrm{BR}_1(b^*)$.
If $a'' \neq a^*$, then $b^* \notin \textrm{BR}_2(a'')$. This is because $u_1(a'',b^*)>u_1(a^*,b^*)$, and $b^*$ best replying against $a''$ implies that committing to $a''$ yields player $1$ a strictly higher payoff, contradicting the definition of $a^*$.
Let $p^*$ be the largest $p \in [0,1]$ such that:
\begin{equation*}
\{b^*\} \neq \textrm{BR}_2 \big(     p a_1^* +(1-p) a'' \big).
\end{equation*}
This suggests the existence of $b'' \neq b^*$ such that $b'' \in \textrm{BR}_2 \big(     p^* a_1^* +(1-p^*) a'' \big)$. My first requirement is:
\begin{equation}\label{B.1}
    u_1(a'',b'') \geq u_1(a^*,b'').
\end{equation}
Recall that $(a',b')$ is player $1$'s worst stage-game Nash equilibrium, which is strict under Assumption \ref{Ass1}. When $a' \neq a^*$,
$b' \notin \textrm{BR}_2(a^*)$. This is because otherwise, $b' \in \textrm{BR}_2(a^*)$, and given that $b^* \in \textrm{BR}_2(a^*)$, Assumption \ref{Ass1} implies that $b^*=b'$. As a result, $(a',b^*)$ is a stage-game Nash Equilibrium, and Assumption \ref{Ass1} implies that $u_1(a',b^*)>u_1(a^*,b^*)$, i.e., player $1$ obtains strictly higher payoff by committing to $a'$ compared to committing to $a^*$. This contradicts the presumption that $a^*$ is player $1$'s pure Stackelberg action.
Let $q^*$ be the smallest $q \in [0,1]$ such that:
 \begin{equation*}
\{b'\} \neq \textrm{BR}_2 \big(     q a_1^* +(1-q) a' \big).
\end{equation*}
This suggests the existence of $b^{**} \neq b'$ such that $b^{**} \in \textrm{BR}_2 \big(     q^* a_1^* +(1-q^*) a' \big)$. My second and third requirements are:
\begin{equation}\label{B.2}
    u_1(a',b^{**}) \geq u_1(a^*,b^{**}),
\end{equation}and
\begin{equation}\label{B.3}
    u_1(a'',b^*)-u_1(a^*,b^*) \geq u_1(a',b^{**})-u_1(a^*,b^{**}).
\end{equation}
I introduce two classes of games, starting from games with strict lack-of-commitment.
\begin{Definition}
$(u_1,u_2)$ is a game with strict lack-of-commitment if $(a^*,b^*)$ is not a Nash Equilibrium, and players' payoffs
satisfy (\ref{B.1}), (\ref{B.2}), and (\ref{B.3}).
\end{Definition}
Under Assumption \ref{Ass1}, the requirement that $(a^*,b^*)$ is not a Nash Equilibrium implies that $a^* \neq a'$ and $a^* \neq a''$, i.e., the Stackelberg action is \textit{strictly suboptimal} for player $1$ both when player $2$ plays her Stackelberg best reply $b^*$ and when she plays her Nash equilibrium action $b'$. Actions $a'$ and $a''$ are player $1$'s best replies to these player $2$'s actions, where $a'$ and $a''$ can potentially coincide.
Inequality (\ref{B.1}) requires that player $1$ benefits from deviating to $a''$ not only when player $2$ plays $b^*$, but also when player $2$ plays $b''$, her best reply when she faces uncertainty about whether player $1$ will play $a^*$ or $a''$.
Inequality (\ref{B.2}) requires that player $1$ benefits from deviating to $a'$ not only when player $2$ plays $b'$, but also when player $2$ plays $b^{**}$, her best reply when she faces uncertainty about whether player $1$ will play $a^*$ or $a'$.
Inequality (\ref{B.3}) requires that player $1$'s benefit from cheating is larger when player $2$ plays her Stackelberg best reply.
Lemma \ref{LB.1} shows that games with strict lack-of-commitment contains games with monotone-supermodular payoffs.
\begin{Lemma}\label{LB.1}
When a game's stage-game payoffs are monotone-supermodular, then it is a game with strict lack-of-commitment.
\end{Lemma}
\begin{proof}[Proof of Lemma B.1:]
Since $u_1(a,b)$ is strictly decreasing in $a$, we have $a'=a''=\underline{a}$.
Since $a^* \neq \underline{a}$, we have $a^* \neq a'$ and $a^* \neq a''$.
Since $a^* \succ a'$ and $a^* \succ a''$, we obtain (\ref{B.1}) and (\ref{B.2}).
By construction, $p^* \geq q^*$. Since $u_2(a,b)$ has strictly increasing differences in $a$ and $b$, we have $b^* \succeq b^{**}$.
Given that $a^* \succ \underline{a}=a'=a''$ and $u_1(a,b)$ has strictly decreasing differences in $a$ and $b$, which yields (\ref{B.3}).
\end{proof}
Next, I define generalized coordination games.
\begin{Definition}
$(u_1,u_2)$ is a general coordination game if  $(a^*,b^*)$ is a Nash Equilibrium.
\end{Definition}
When player $1$'s Stackelberg outcome is a Nash Equilibrium, either $a^*=a'$ or $a^*=a''$ or both. This can be further categorized into two subclasses: trivial games in which the pure Stackelberg outcome coincides with the worst pure strategy Nash equilibrium (for example, prisoner's dilemma), games that have at least two pure-strategy Nash Equilibria. When $(a^*,b^*)$ is a Nash Equilibrium, it must be player $1$'s favorite Nash Equilibrium, while $(a',b')$ is player $1$'s least favorite Nash Equilibrium. This includes for example, battle of sexes, chicken games, and common interest coordination games in which different equilibria can be Pareto ranked.
Theorem 3' generalizes Theorem \ref{Theorem3} to games with strict lack-of-commitment and generalized coordination games.
\begin{Theorem3}
If the monitoring structure $\mathcal{N}$ satisfies Assumption \ref{Ass4} and the stage game satisfies Assumption \ref{Ass1}, and is
either a game with strict lack-of-commitment or a generalized coordination game,
then there exists $\overline{\pi}_0 \in (0,1)$,
such that for every $\pi_0 \in (0,\overline{\pi}_0)$ and $\delta$ large enough,
there exists $(\sigma_1^{\delta},\sigma_2^{\delta}) \in \textrm{SE}(\delta,\pi_0,\mathcal{N})$, such that:
\begin{equation*}
     \mathbb{E}_1^{(\sigma_1^{\delta},\sigma_2^{\delta})} \Big[   \sum_{t=0}^{\infty} (1-\delta) \delta^t u_1(a_t,b_t) \Big] =\underline{v}_1.
\end{equation*}
\end{Theorem3}

\subsection{Proof of Theorem 3': Games with Strict Lack-of-Commitment}\label{subB.2}
The constructed equilibrium consists of three phases: a \textit{reputation building phase}, a \textit{reputation maintenance phase}, and a \textit{punishment phase}, which
\textit{depends only} on the history of player $2$'s actions that is commonly observed by both players.
In period $t$,
\begin{itemize}
  \item play is in the reputation building phase if $t=0$ or $(b_0,...,b_{t-1})=(b',...,b')$;
  \item play is in the reputation maintenance phase if first, there exists $s \leq t-1$ such that $b_s=b^{**}$, and second,
  $(b_{s^*+1},...,b_{t-1})=(b^*,...,b^*)$ where
 $s^*$ is the smallest $s \in \mathbb{N}$ such that $b_s=b^{**}$.
  \item play is in the punishment phase if first, there exists $s \leq t-1$ such that $b_s=b^{**}$, and second,
  $(b_{s^*+1},...,b_{t-1}) \neq (b^*,...,b^*)$ where
 $s^*$ is the smallest $s \in \mathbb{N}$ such that $b_s=b^{**}$.
\end{itemize}
Play starts from the reputation building phase, and eventually ends up in the reputation maintenance phase or the punishment phase. Different from the construction in Theorem \ref{Theorem1}, the punishment phase is reached with strictly positive probability due to private monitoring and private learning. In what follows, I describe players' strategies and verify their incentive constraints in each of the three phases. By the end of this proof, I verify the promise keeping condition for player $1$.

\paragraph{Punishment Phase:} At every punishment phase history, player $1$ plays $a'$ and player $2$ plays $b'$. As will become clear after describing the reputation maintenance phase, play never reaches the punishment phase conditional on player $1$ being committed.
This
implies the existence of an assessment that is consistent with the equilibrium strategy profile and attaches probability $1$ to the strategic type at every punishment-phase history. This verifies players' incentive constraints in the punishment phase.

\paragraph{Reputation Maintenance Phase:} Let $s^*$ be the smallest $s$ such that $b_s=b^{**}$. In period $t \geq s^* +2$,
\begin{itemize}
  \item If $a_{t-1}=a^*$, then strategic-type player $1$ plays $a^*$.
  \item[] If $a_{t-1} \neq a^*$, then strategic-type player $1$ plays $p^* a_1^* +(1-p^*) a''$.
  \item If $t-1 \notin N_t$ or $a_{t-1} = a^*$, then player $2$ plays $b^*$.
  \item[] If $t-1 \in N_t$ and $a_{t-1} \neq a^*$, then player $2$ plays $\widetilde{\beta}_t b^* +(1-\widetilde{\beta}_t) b''$.
Let $\beta$ be the probability with which player $2$ plays $b^*$ conditional on $a_{t-1} \neq a^*$ but unconditional on realization of $\mathcal{N}_t$. It is also the unconditional probability that play remains in the reputation maintenance phase in period $t+1$ given that $a_{t-1} \neq a^*$. One can compute $\beta$ and
player $1$'s continuation value in period $t$ when $a_{t-1} \neq a^*$, denoted by $V_1$,
by solving the following system of quadratic equations:
   \begin{equation}\label{B.4}
    V_1=(1-\delta) u_1(a^*,\beta b^* +(1-\beta)b'') +\delta \beta u_1(a^*,b^*)
    +\delta (1-\beta_t) u_1(a',b'),
    \end{equation}
    \begin{equation}\label{B.5}
   V_1 = (1-\delta) u_1(a'',\beta b^* +(1-\beta)b'') +\delta \beta V_1
    +\delta (1-\beta_t) u_1(a',b').
  \end{equation}
 To understand (\ref{B.4}) and (\ref{B.5}), note that
 $u_1(a^*,b^*)$ is player $1$'s continuation value in period $t+1$ when $a_t=a^*$ and play remains in the reputation maintenance phase, and $u_1(a',b')$ is player $1$'s continuation value in period $t+1$ when play reaches the punishment phase.
Player $2$'s mixing probability conditional on $t-1 \in N_t$ and $a_{t-1} \neq a^*$ satisfies:
      \begin{equation}\label{B.beta}
      1-  \beta = (1-\widetilde{\beta}_t) \Pr(t-1 \in N_t).
      \end{equation}
 Since $\Pr(t-1 \in N_t)$ is uniformly bounded from below by $\gamma>0$, $\widetilde{\beta}_t \in (0,1)$ when $\beta>1-\gamma$. Lemma \ref{LB.2} verifies that $\beta$ converges to $1$ as $\delta \rightarrow 1$, i.e., $\beta>1-\gamma$ when $\delta$ is large enough.
\end{itemize}
In period $s^* +1$,
\begin{itemize}
  \item If $a_{s^*}=a^*$,
  then strategic-type player $1$ plays $a^*$.
  \item[] If $a_{s^*} \neq a^*$ and
  \begin{equation}\label{B.6}
      \xi_{s^*} > \overline{\xi} \equiv \frac{u_1(a',b^{**})-u_1(a^*,b^{**})}{\displaystyle u_1(a'',b^*)-u_1(a^*,b^*) +\frac{1-\beta}{\beta} \Big(u_1(a'',b'')-u_1(a^*,b'')\Big)},
  \end{equation}
  then strategic-type player $1$ plays $a^*$.
  \item[] If $a_{s^*} \neq a^*$ and $\xi_{s^*} \leq \overline{\xi}$,
  then strategic-type player $1$ plays $p^* a_1^* +(1-p^*) a''$.
  \item If $s^* \notin N_{s^*+1}$, or $a_{s^*} = a^*$, or $\xi_{s^*}$ satisfies (\ref{B.6}), then player $2$ plays $b^*$.
  \item[] If $s^* \in N_{s^*+1}$, $a_{s^*} \neq a^*$, and $\xi_{s^*}$ does not satisfy (\ref{B.6}), then player $2$ plays $\widetilde{\beta} b^* +(1-\widetilde{\beta}) b''$, where $\widetilde{\beta}$ can be solved via (\ref{B.4}), (\ref{B.5}) and (\ref{B.beta}).
\end{itemize}
To verify players' incentive constraints in this phase, I only need to verify their incentives from period $s^*+2$ and onwards. This is because the equilibrium play in period $s^*+1$ is a randomization between the two automaton states from period $s^*+2$ and onwards.

I start from verifying player $2$s' incentives. According to the definition of $p^*$ in Appendix \ref{subB.1}, player $2$ is indifferent between $b^*$ and $b''$ when her belief about player $1$'s action is $p^* a_1^* +(1-p^*) a''$.
In the constructed strategy for player $2$s in the reputation maintenance phase,
\begin{itemize}
  \item Player $2$ randomizes between $b^*$ and $b'$ in period $t$
only when she has observed $a_{t-1} \neq a^*$, after which her posterior belief attaches probability $0$ to the commitment type, and therefore, believes that player $1$'s action is $p^* a_1^* +(1-p^*) a''$.
  \item Player $2$ plays $b^*$ at other histories. This is incentive compatible since $b^*$ best replies against $p a_1^* +(1-p) a''$
for every $p \in [p^*,1]$, and
given that the strategic type plays $p^* a_1^* +(1-p^*) a''$,
the unconditional probability with which player $1$ plays $a^*$ is between $p$ and $1$.
\end{itemize}
Next, I verify player $1$'s incentives. I show that first, when $a_{t-1}=a^*$, player $1$ has an incentive to play $a^*$. This requires:
\begin{equation*}
    u_1(a^*,b^*) \geq \max_{a \neq a^*} \Big\{(1-\delta) u_1(a,b^*)
    +\delta
    V_1
    \Big\},
\end{equation*}
where $V_1$ solves (\ref{B.4}) and (\ref{B.5}). Intuitively, $V_1$
is player $1$'s continuation value in the reputation maintenance phase conditional on his previous period action is not $a^*$.
Given that $a''$ is player $1$'s stage-game best reply against $b^*$, the above inequality reduces to:
\begin{equation}\label{B.7}
    u_1(a^*,b^*)-V_1 \geq \frac{1-\delta}{\delta} (u_1(a'',b^*)-u_1(a^*,b^*)).
\end{equation}
Deduct the RHS of (\ref{B.4}) from the RHS of (\ref{B.5}), we obtain:
\begin{equation}\label{B.8}
    \delta \beta (u_1(a^*,b^*)-V_1)
    =(1-\delta) \beta (u_1(a'',b^*)-u_1(a^*,b^*))
    +(1-\delta) (1-\beta) (u_1(a'',b'')-u_1(a^*,b''))
\end{equation}
Equation (\ref{B.8}) implies that $\beta>0$. This is because otherwise, $u_1(a'',b'')-u_1(a^*,b'')=0$, and given the presumption that $a^* \neq a''$, this violates Assumption \ref{Ass1}. According to (\ref{B.1}),
$u_1(a'',b'')-u_1(a^*,b'') \geq 0$, and therefore, (\ref{B.8}) implies (\ref{B.7}).

Next, I show that when $a_{t-1} \neq a^*$, player $1$ has an incentive to mix between $a^*$ and $a''$ in period $t$.
Since $\mathcal{N}_t$ is independent of $\{\mathcal{N}_{s}\}_{s=0}^{t-1}$,
and $\beta$ is the probability with which player $2$ plays $b^*$ in period $t$ \textit{conditional on} $a_{t-1} \neq a^*$ but \textit{unconditional on} the realization of $\mathcal{N}_t$,
player $1$ believes that player $2$ plays $b^*$ with probability $\beta$ conditional on $a_{t-1} \neq a^*$.
Equations (\ref{B.4}) and (\ref{B.5})
imply that player $1$ is indifferent between $a^*$ and $a''$. What remains to be shown is that player $1$ prefers $a''$ to actions other than $a^*$ and $a''$. This hinges on the following lemma, implying that $\beta$ is close to $1$ when $\delta$ is close to $1$.
\begin{Lemma}\label{LB.2}
For every $\gamma \in (0,1)$, there exists $\underline{\delta} \in (0,1)$, such that
for every $\delta>\underline{\delta}$, there exists
$\beta \in (1-\gamma,1)$ that solves (\ref{B.4}) and (\ref{B.5}).
\end{Lemma}
The proof requires some algebra and is relegated to Appendix \ref{subB.4}.
Since $a''$ best replies against $b^*$, Assumption \ref{Ass1} implies that $a''$ is a strict best reply against $b^*$. Lemma \ref{LB.2} implies that there exists $\delta$ large enough such that $a''$ best replies against $\beta a^* +(1-\beta)a''$. This suggests that player $1$ receives higher payoff by playing $a''$ compared to actions other than $a^*$ and $a''$.

\paragraph{Reputation Building Phase:} First, I describe player $2$s' equilibrium strategy and verifies strategic-type player $1$'s incentive
to mix between $a^*$ and $a'$ at every private history of the reputation building phase. Second, I describe strategic-type player $1$'s equilibrium strategy, and verifies player $2$s' incentives to mix between $b^{**}$ and $b'$ at every private history of the reputation building phase.

\paragraph{Player $2$s' Strategy:} Player $2$ plays $b'$ if (1) $t=0$, or (2) $t-1 \notin N_t$, or (3) $t-1 \in N_t$ but $a_{t-1} \neq a^*$.
Player $2$ plays $\widetilde{\rho}_t b^{**} +(1-\widetilde{\rho}_t) b'$ if $t-1 \in N_t$ and $a_{t-1}=a^*$.
Let $\rho$ be the probability with which player $2$ plays $b^{**}$ conditional on $a_{t-1}=a^*$ but unconditional on the realization of $\mathcal{N}_t$. This probability $\rho$ is pinned down by the following equation:
\begin{equation}\label{B.12}
        V_1'
        =
        (1-\delta) u_1(a^*,\rho b^{**}+(1-\rho)b')
        +\delta \rho u_1(a^*,b^*)
        +\delta (1-\rho) V_1'.
\end{equation}
where $V_1'$ is given by:
\begin{equation}\label{B.13}
u_1(a',b')= (1-\delta) u_1(a^*,b') +\delta V_1'.
\end{equation}
Player $2$'s mixing probability conditional on
$t-1 \in N_t$ and $a_{t-1}=a^*$, which is $\widetilde{\rho}_t$,
satisfies:
\begin{equation}\label{B.11}
        \rho=\widetilde{\rho}_t \Pr(t-1 \in N_t),
\end{equation}
Since $\Pr(t-1 \in N_t)$ is uniformly bounded from below by $\gamma$, $\widetilde{\rho}_t \in (0,1)$ as long as $\rho \in (0,\gamma)$.  Lemma \ref{LB.3} shows that $\rho \rightarrow 0$ as $\delta \rightarrow 1$.  In what follows, I show that such mixing probabilities give strategic-type player $1$ an incentive to mix between $a^*$ and $a'$ at every reputation-building phase history.

First, $u_1(a',b')$ is player $1$'s continuation value in the reputation building phase when $t=0$ or $a_{t-1} \neq a^*$.
Equation (\ref{B.13}) implies that player $1$ is indifferent between $a^*$ and $a'$ when $t=0$ or $a_{t-1} \neq a^*$, and moreover, given that $a'$ is player $1$'s stage-game best reply against $b'$, $a'$ yields strictly higher payoff compared to actions other than  $a^*$ and $a'$.

Next, recall that $u_1(a^*,b^*)$ is player $1$'s continuation payoff in the reputation maintenance phase conditional on
player $1$ playing $a^*$ in the period before moving to the reputation maintenance phase; and $\overline{\xi} V_1 + (1-\overline{\xi}) u_1(a^*,b^*)$ is player $1$'s continuation payoff in the reputation maintenance phase conditional on
player $1$ not playing $a^*$ in the period before moving to the reputation maintenance phase, where
$V_1$ solves (\ref{B.4}) and (\ref{B.5}) and
$\overline{\xi}$ is given by (\ref{B.6}).
At every reputation-building phase history with $a_{t-1} =a^*$, player $1$'s expected payoff from playing $a^*$ is given by the RHS of (\ref{B.12}). His expected payoff from playing $a'$ is:
\begin{equation}\label{B.14}
   (1-\delta) u_1(a', \rho b^{**} +(1-\rho) b') +\delta \rho \Big( \overline{\xi} V_1 + (1-\overline{\xi}) u_1(a^*,b^*) \Big)
   +\delta (1-\rho) u_1(a',b').
\end{equation}
Let $V_1^{**} \equiv \overline{\xi} V_1 + (1-\overline{\xi}) u_1(a^*,b^*)$.
Subtracting  the RHS of (\ref{B.12}) from (\ref{B.14}), we obtain:
\begin{equation}\label{B.16}
    (1-\delta)\rho (u_1(a^*,b^{**})-u_1(a',b^{**}))
    +(1-\delta) (1-\rho) (u_1(a^*,b')-u_1(a',b'))
    +\delta \rho (u_1(a^*,b^*)-V_1^{**})
    +\delta (1-\rho) (V_1'-u_1(a',b'))
\end{equation}
According to (\ref{B.13}), (\ref{B.16}) reduces to:
\begin{equation*}
    (1-\delta)\rho (u_1(a^*,b^{**})-u_1(a',b^{**}))
    +\delta \rho (u_1(a^*,b^*)-V_1^{**}),
\end{equation*}
which equals $0$ according to (\ref{B.6}).
This suggests that (\ref{B.14}) equals the RHS of (\ref{B.12}), and that
the strategic-type player $1$ is indifferent between $a^*$ and $a'$. To show that player $1$ strictly prefers $a'$ to actions other than $a^*$ and $a'$, I establish the following lemma:
\begin{Lemma}\label{LB.3}
For every $\gamma \in (0,1)$, there exists $\underline{\delta} \in (0,1)$, such that
for every $\delta>\underline{\delta}$, there exists
$\rho \in (0,\gamma)$ that solves (\ref{B.12}) and (\ref{B.13}).
\end{Lemma}
\begin{proof}[Proof of Lemma B.3:]
According to (\ref{B.13}), $V_1'$ converges to $u_1(a',b')$ as $\delta \rightarrow 1$.
According to (\ref{B.12})
\begin{equation}\label{B.17}
   (1-\delta+\delta\rho)     V_1'
        =
        (1-\delta) u_1(a^*,\rho b^{**}+(1-\rho)b')
        +\delta \rho u_1(a^*,b^*).
\end{equation}
Suppose toward a contradiction that there exists a sequence of $\{\delta_n\}_{n =1}^{\infty}$ with $\lim_{n \rightarrow \infty} \delta_n =1$ such that $\lim_{n \rightarrow \infty} \rho_n =\rho^*>0$.
Then the LHS of (\ref{B.17}) converges to $\rho^* V_1'$ while the RHS converges to $\rho^* u_1(a^*,b^*)$. Since $V_1' \rightarrow u_1(a',b')$ and $u_1(a',b')<u_1(a^*,b^*)$ in games with strict lack-of-commitment. This yields a contradiction.
\end{proof}
Since $a'$ is a strict best reply against $b'$, there exists $\overline{\varepsilon}>0$ such that $a'$ best replies against $\varepsilon b^{**}+(1-\varepsilon) b'$ for every $\varepsilon \in (0,\overline{\varepsilon})$.
Lemma \ref{LB.3} implies that when $\delta$ is close to $1$, $a'$ is player $1$'s strict best reply against
$\rho b^{**}+(1-\rho) b'$, i.e.,  $a'$ yields player $1$ a strictly higher payoff compared to actions
other than
$a^*$ and $a'$.

\paragraph{Player $1$'s Strategy:} I start from the following Lemma:
\begin{Lemma}\label{LB.4}
For every $K \in \mathbb{N}$, there exists  $M \in \mathbb{N}$, such that
\begin{equation}\label{B.9}
   2^K \sum_{j=0}^{K} {n \choose j} < 2^n, \quad \textrm{for every } n \geq M.
\end{equation}
\end{Lemma}
\begin{proof}[Proof of Lemma B.4:]
To start with, $\sum_{j=0}^{n} {n \choose j}=2^n$ for every $n \in \mathbb{N}$.
Moreover, ${n \choose j}$ is increasing in $j$ when $j<n/2$ and is decreasing in $j$ when $j > n/2$.
This suggests that for every $m \geq 2$ and $n > mK$,
\begin{equation*}
    \sum_{j=0}^{K} {n \choose j} \leq \frac{2^n}{m}
\end{equation*}
Let $M \equiv 2^K K$, we have:
\begin{equation*}
    2^K \sum_{j=0}^{K} {n \choose j} < \frac{2^{K+n}}{2^K}=2^n.
\end{equation*}
\end{proof}
Let $\eta \in (0,1/2)$ be a small enough real number which will be determined by the end of the construction (Lemma \ref{LB.5}).
Pick $\overline{\pi}_0(\eta) \in (0,1)$ small enough such that:
\begin{equation}\label{B.10}
    \frac{\overline{\pi}_0(\eta)}{1-\overline{\pi}_0(\eta)} \Big(\frac{1}{\eta q^*} \Big)^M < \frac{\eta q^*}{1-\eta q^*}.
 \end{equation}

Recall the definitions of $a'$, $b'$, $b^{**}$, and $q^*$ in Appendix \ref{subB.1}.
Player $1$'s strategy in the reputation building phase depends on \textit{calendar time}, and in particular, the comparison between $t$ and $M$.
Let $\boldsymbol{\widetilde{\pi}_t} \in \Delta (\Omega)$ be the posterior belief of a hypothetical observer who shares the same prior belief as player $2$s, but can observe the entire history of actions and public randomization devices, i.e., he observes $\{a_0,...,a_{t-1},b_0,...,b_{t-1},\xi_0,...,\xi_{t}\}$.
Let $\widetilde{\pi}_t$ be the probability
$\boldsymbol{\widetilde{\pi}_t}$ attaches to the commitment type. Player $1$ can compute $\widetilde{\pi}_t$ based on his private history, but player $2$ cannot.

At every reputation-building phase history $h^t$ with $t \leq M$, strategic-type player $1$ mixes between $a^*$ and $a'$, with the probability of playing $a^*$ being $q(h_1^t)$, satisfying:
\begin{equation}\label{B.18}
 (1-\widetilde{\pi}_t)   q(h_1^t) +\widetilde{\pi}_t =q^*.
\end{equation}
i.e., player $2$ is indifferent between $b^{**}$ and $b'$ when she observes the entire history of actions and public randomization devices.

I show that when $\pi_0 < \overline{\pi}_0 (\eta)$, we have
$\widetilde{\pi}_t <\eta q^*$ for every $\{a_0,...,a_{t-1},b_0,...,b_{t-1},\xi_0,...,\xi_{t}\}$ with $t \leq M$.
This is because $\widetilde{\pi}_t$ is bounded from above by an observer's belief who observes
$\{a_0,...,a_{t-1}\}=\{a^*,a^*,...,a^*\}$. According to (\ref{B.18}),
$q(h_1^t) > \eta q^*$ when
$\widetilde{\pi}_t< \eta q^*$, and therefore, $\widetilde{\pi}_{t+1} \leq \frac{\widetilde{\pi}_t}{\eta q^*}$.
Applying (\ref{B.10}), one can then show inductively that $\widetilde{\pi}_t< \eta q^*$ for every $t \leq M$.
The above conclusion also suggests that $q(h_1^t) > \eta q^*$ for every $h_1^t$ with $t \leq M$.

At every reputation-building phase history $h^t$ with $t > M$, strategic-type player $1$ mixes between $a^*$ and $a'$.
The probability with which he plays $a^*$ depends on his private history only through $\chi^t \equiv \{\chi_0,...,\chi_{t-1}\} \in X^t \equiv \{0,1\}^{t}$,
with
\begin{displaymath}
\chi_s = \left\{ \begin{array}{ll}
1 & \textrm{if } a_s=a^* \\
0 & \textrm{if } a_s \neq a^*.
\end{array} \right.
\end{displaymath}
Let $q(\chi^t)$ be the probability with which strategic-type player $1$ plays $a^*$, and let $\mathbf{q}_t \equiv \{q(\chi^t)\}_{\chi^t \in \{0,1\}^t}$. For every $h_2^t$, let $\kappa(h_2^t) \in \Delta(X_t)$ be player $2$'s belief about $\chi^t$ conditional on player $1$ being the strategic type. For every $t >M$,
recall that $\pi(h_2^t)$ is the probability player $2$'s posterior belief attaches to the commitment type when her private history is $h_2^t$.
Let $\mathbf{q}_t$ be such that:
\begin{equation}\label{B.19}
 \pi(h_2^t)+(1-\pi(h_2^t))  \kappa(h_2^t) \cdot \mathbf{q}_t   = q^* \textrm{ for every } h_2^t \in \mathcal{H}_2^t.
\end{equation}
When (\ref{B.19}) is satisfied, player $2$ believes that player $1$ plays $a^*$ with probability $q^*$ and $a'$ with probability $1-q^*$ at private history $h_2^t$, and is therefore, indifferent between $b^{**}$ and $b'$.
I show the following lemma, which verifies player $2$'s incentive constraints in the reputation building phase.
\begin{Lemma}\label{LB.5}
There exists $\eta \in (0,1/2)$ such that for every $t >M$,
if $q(\chi^s) \in[ \eta q^*,1-\eta q^*]$ for every $\chi^s \in \{0,1\}^s$ with $s \leq t-1$, then there exists $\mathbf{q}_t \in [0,1]^{2^t}$ that solves (\ref{B.19})
and satisfies $q(\chi^t) \in [\eta q^*, 1-\eta q^*]$ for every $\chi^t \in \{0,1\}^t$.
\end{Lemma}
\begin{proof}[Proof of Lemma B.5:]
First, I show that linear system (\ref{B.19}), which is equivalent to
\begin{equation*}
    \kappa(h_2^t) \cdot \mathbf{q}_t = \frac{q^*-\pi(h_2^t)}{1-\pi(h_2^t)}\textrm{ for every } h_2^t \in \mathcal{H}_2^t
\end{equation*}
admits a solution. To start with, it is without loss to focus on $h_2^t$ with $|N_t|=K$. This is because for every $|N_t'|<K$, there exists
$|N_t''|=K$ such that an agent
who observes $\{a_s\}_{s \in N_t''}$
has a Blackwell more informative information structure compared to an agent
who observes $\{a_s\}_{s \in N_t'}$. As a result, (\ref{B.20}) is satisfied for all
$h_2^t$ with $|N_t|=K$ implies that (\ref{B.20}) is satisfied for all
$h_2^t$ with $|N_t|\leq K$.

Second, the definition of $M$ suggests that $2^t> 2^K {t \choose K}$ for every $t \geq M$. This suggests that linear system:
\begin{equation}\label{B.20}
      \kappa(h_2^t) \cdot \mathbf{q}_t = \frac{q^*-\pi(h_2^t)}{1-\pi(h_2^t)} \textrm{ for every } h_2^t \in \mathcal{H}_2^t \textrm{ with } |N_t|=K
\end{equation}
is underdetermined. An important observation is that the following set of vectors:
\begin{equation*}
    \{\kappa (h_2^t)\}_{h_2^t \in \mathcal{H}_2^t \textrm{ with } |N_t|=K}
\end{equation*}
are convex independent.
This is because (1) when $s \in N_t$,
player $2$ knows the value of $\chi_s$ and (2)
when $s \notin N_t$, player $2$'s belief about $\chi_s$ is not degenerate.
Let $\boldsymbol{\kappa}_t$ be the coefficient matrix of linear system (\ref{B.20}).
Convex independence of $\kappa(h_2^t)$ suggests
that the rank of the  $\boldsymbol{\kappa}_t$ is ${t \choose K}$.
Let
\begin{equation}\label{B.21}
    \boldsymbol{\widetilde{q}_t} \equiv \Big\{\frac{q^*-\pi(h_2^t)}{1-\pi(h_2^t)}\Big\}_{h_2^t \in \mathcal{H}_2^t \textrm{ with } |N_t|=K}.
\end{equation}
Since (\ref{B.20}) is underdetermined, the rank of the augmented matrix $(\boldsymbol{\kappa}_t,  \boldsymbol{\widetilde{q}_t})$
is also ${t \choose K}$. The Rouch\'{e}-Capelli theorem implies that (\ref{B.20}) admits at least one solution.

Third, I show there exists a solution where each entry of $\mathbf{q}_t$ belongs to the interval $[\eta q^*, 1-\eta q^*]$. This is because the RHS of (\ref{B.20}) converges to $q^*$ as $\pi(h_2^t) \rightarrow 0$, and  linear system (\ref{B.20})
admits a solution $\mathbf{q}_t=(q^*,...,q^*)$
when $\pi(h_2^t)=0$. Proposition 10 in Fefferman and Koll\'{a}r (2013) suggests that the solution correspondence is continuous with respect to $\pi(h_2^t)$. This suggests the existence of $\eta \in (0,1/2)$ such that for every $\pi(h_2^t)< \eta q^*$, there exists a solution to (\ref{B.20}) $\mathbf{q}_t$ such that
\begin{equation*}
    ||\mathbf{q}_t-(q^*,...,q^*)||_{L^{\infty}}< \eta.
\end{equation*}
When $\eta$ is small enough, all entries of $\mathbf{q}_t$ belong to the interval $[q^* \eta ,1-q^* \eta]$.

Fourth, I show by induction that $\pi(h_2^t)< \eta q^*$. In period $t=M$, since the strategic-type player $1$ plays $a^*$ with probability at least $\eta q^*$ at every history of the reputation-building phase before period $M$ and
player $2$'s prior belief attaches probability no more than $\overline{\pi}(\eta)$ to the commitment type, her posterior belief attaches probability at most $\eta \pi^*$ to the commitment type given inequality (\ref{B.10}). The third step then implies the existence of $\mathbf{q}_t$ such that $q(\chi^t) \in [\eta q^*, 1-\eta q^*]$ for every $\chi^t \in \{0,1\}^t$.
Suppose
$\pi(h_2^t)< \eta q^*$ for every $h_2^t \in \mathcal{H}_2^t$ and every $t \leq T$. The same reasoning implies that $\pi(h_2^{T+1})< \eta q^*$ for every $h_2^{T+1} \in \mathcal{H}_2^{T+1}$ given that each player $2$ observes at most $K$ realizations among $\{a_t\}_{t=0}^T$. This also suggests the existence of $\mathbf{q}_{T+1}$ with
$q(\chi^{T+1}) \in [\eta q^*, 1-\eta q^*]$ for every $\chi^{T+1} \in \{0,1\}^{T+1}$.
\end{proof}

\paragraph{Promising Keeping Constraint:} I conclude the proof by verifying player $1$'s promise keeping constraint, i.e., the continuation play delivers strategic-type player $1$ his continuation value in period $0$, which is $u_1(a',b')$. Since strategic-type player $1$ plays $a^*$ with probability at least $\eta q^*$ in the reputation building phase, and conditional on $a_{t-1}=a^*$, play transits to the reputation maintenance phase with probability $\rho$, the equilibrium play belongs to the reputation maintenance phase or punishment phase with probability $1$ as $t \rightarrow \infty$. This suggests that on the equilibrium path, play either converges to $(a^*,b^*)$ in every period, or $(a',b')$ in every period, and player $1$'s continuation value in those cases are $u_1(a^*,b^*)$ and $u_1(a',b')$, respectively. This verifies the promise keeping constraints.
\subsection{Proof of Theorem 3': Generalized Coordination Games}\label{subB.3}
First, consider the trivial case in which there exists a unique pure strategy Nash Equilibrium, i.e.,
$a'=a^*$. Player $1$'s payoff is $u_1(a',b')$ in an equilibrium where player $1$ plays $a^*$ at every history and player $2$ plays $b^*$ at every history.

Next, consider the nontrivial case in which $a''=a^*$ but $a' \neq a^*$, i.e., the Stackelberg outcome is a pure-strategy
Nash Equilibrium in the stage game, but there also exists another stage-game Nash Equilibrium which results in strictly lower payoff for player $1$.

I construct an equilibrium that
consists of three phases: a \textit{reputation building phase}, a \textit{reputation maintenance phase}, and a \textit{punishment phase}, which
\textit{depends only} on the history of player $2$'s actions that is commonly observed by both players. I only highlight the differences between this construction and the one in Appendix \ref{subB.2}, in order to avoid repetition.
In period $t$,
\begin{itemize}
  \item play is in the reputation building phase if $t=0$ or $(b_0,...,b_{t-1})=(b',...,b')$;
  \item play is in the reputation maintenance phase if (1) there exists $s \leq t-1$ such that $b_s=b^{**}$, and (2)
  $b_{s^*+1} =b^*$ where $s^*$ is the smallest $s \in \mathbb{N}$ such that $b_s=b^{**}$.
  \item play is in the punishment phase if (1) there exists $s \leq t-1$ such that $b_s=b^{**}$, and (2)
  $b_{s^*+1} \neq b^*$ where $s^*$ is the smallest $s \in \mathbb{N}$ such that $b_s=b^{**}$.
\end{itemize}
Play starts from the reputation building phase, and eventually ends up in the reputation maintenance phase or the punishment phase.
Different from the construction in Appendix \ref{subB.2}:
\begin{enumerate}
  \item The set of reputation maintenance phase histories is larger in the current construction. In particular, player $2$ only checks her predecessor's action in the next period after $b^{**}$ occurs, and as long as it is $b^*$, play remains in the reputation maintenance phase regardless of players' actions after period $s^*+1$.
  \item Players' strategies differ in the reputation maintenance phase, and player $2$'s strategy is different in the reputation building phase.
\end{enumerate}
Players' strategies in the punishment phase and player $1$'s strategy in the reputation building phase remain the same as in Appendix \ref{subB.2}, which I omit to avoid repetition.
The differences are in players' strategies in the reputation maintenance phase.

\paragraph{Reputation Maintenance Phase:}  Let $s^*$ be the smallest $s$ such that $b_s=b^{**}$. At every history of the reputation maintenance phase with $t \geq s^* +2$, strategic-type player $1$ plays $a^*$ and player $2$ plays $b^*$. In period $s^* +1$,
\begin{itemize}
  \item If $a_{s^*}=a^*$, then strategic-type player $1$ plays $a^*$ and player $2$ plays $b^*$.
  \item[] If $a_{s^*} \neq a^*$ and
  \begin{equation}\label{B.22}
      \xi_{s^*} > \overline{\xi}' \equiv \frac{1-\delta}{\delta} \cdot \frac{u_1(a',b')-u_1(a^*,b')}{u_1(a^*,b^*)-u_1(a',b')}.
  \end{equation}
  then strategic-type player $1$ plays $a^*$.
  \item[] If $a_{s^*} \neq a^*$ and $\xi_{s^*} \leq \overline{\xi}'$,
  then strategic-type player $1$ plays $a'$ and player $2$ plays $b'$.
\end{itemize}
\paragraph{Reputation Building Phase:} Player $1$'s strategy remains the same as Appendix \ref{subB.2}. Player $2$s' strategy differs
in terms of $\rho$, the probability of playing $b^{**}$ in period $t$ conditional on $a_{t-1}=a^*$ but unconditional on the realization of $\mathcal{N}_t$. This $\rho$ is pinned down by the following equality:
\begin{equation}\label{B.23}
    V_1'=(1-\delta)u_1(a^*,\rho b^* +(1-\rho)b') +\delta \rho u_1(a^*,b^*) +\delta (1-\rho) V_1',
\end{equation}
with
\begin{equation*}
    u_1(a',b')=(1-\delta) u_1(a^*,b')+\delta V_1'.
\end{equation*}
The rest of the construction remains the same.

\subsection{Proof of Lemma B.2}\label{subB.4}
\begin{proof}[Proof of Lemma B.2:]
According to (\ref{B.8}), we have:
\begin{equation*}
    V_1=u_1(a^*,b^*)- \frac{1-\delta}{\delta} \Big(u_1(a'',b^*)-u_1(a^*,b^*)\Big)
    -\frac{1-\delta}{\delta}\frac{1-\beta}{\beta} \Big(u_1(a'',b'')-u_1(a^*,b'')\Big).
\end{equation*}
Plugging this into (\ref{B.4}), we obtain:
\begin{equation*}
    \frac{1}{1-\beta} \underbrace{\Big(u_1(a'',b^*)-u_1(a^*,b^*)\Big)}_{\equiv Y}
    + \frac{1}{\beta} \underbrace{\Big(u_1(a'',b'')-u_1(a^*,b'')\Big)}_{\equiv Z}
    =\underbrace{\frac{\delta}{1-\delta} \Big(
    u_1(a^*,b^*)-(1-\delta) u_1(a^*,b'') -\delta u_1(a',b')
    \Big)}_{\equiv X}.
\end{equation*}
The two roots are given by:
\begin{equation}
    \beta = \frac{X+Z-Y \pm \sqrt{(X+Z-Y)^2-4XZ}}{2X}
\end{equation}
When the stage-game features strict lack-of-commitment (see Definition 1), $X>0$ and $Y,Z \geq 0$.
As $\delta \rightarrow 1$, $X \rightarrow \infty$ while $Y$ and $Z$ remain constant, and therefore,
both roots are real. Moreover,
the larger of the two roots converge to $1$. To show that it is strictly less than $1$, it is equivalent to show that:
\begin{equation*}
    \sqrt{(X+Z-Y)^2-4XZ} < 2X -(X+Z-Y),
\end{equation*}
which is equivalent to:
\begin{equation*}
    X^2+Y^2+Z^2-2XZ-2YZ-2XY < X^2+Y^2+Z^2-2XZ-2YZ+2XY.
\end{equation*}
The last inequality is true given that both $X$ and $Y$ are strictly positive. This establishes the existence of a real root $\beta$ that is strictly less than $1$ and converges to $1$ as $\delta \rightarrow 1$.
\end{proof}

\section{Proof of Theorem 4}\label{secC}
\subsection{Proof of Theorem 4: Statement 1}\label{subC.1}
For every public history $h^t$,
let $g(h^t)$ be the probability that player $2$ plays $b^*$ at $h^t$. Let $g(h^t,\omega^c)$ be the probability that player $2$ plays $b^*$ at $h^t$ conditional on player $1$ is the commitment type. For any public history $h^t \equiv \{a_{\max\{0,t-K\}},...,a_{t-1}\}=\{a^*,....,a^*\}$, I derive a lower bound on:
    $\frac{g(h^t,\omega^c)}{g(h^t)}$
as a function of $g(h^t)$,
or equivalently, an upper bound on
\begin{equation}\label{4.3}
    \frac{1-g(h^t,\omega^c)}{1-g(h^t)}.
\end{equation}
Let $A \equiv \{a^*,a'\}$ and $S \equiv \{s^*,s_1,s_2,...,s_m\}$. Let
$r(h^t)$ be the probability that $a^*$ is played at $h^t$, let $\tau(s_i)(h^t)$ be the probability that signal $s_i$ occurs at $h^t$, and let $p(s_i)(h^t)$ be the posterior probability of $a^*$ conditional on observing $s_i$ at $h^t$.
I suppress the dependence on $h^t$ in order to simplify notation.
Since $\{b^*\} =\textrm{BR}_2(a^*)$ and $|A|=2$, we have the following two implications:
\begin{enumerate}
  \item there exists a cutoff belief $p^* \in (0,1)$ such that player $2$ has a strict incentive to play $b^*$ after observing $s_i$ if and only if $p(s_i) > p^*$.
  \item there exists a constant $C \in \mathbb{R}_+$ such that $1-r \geq C (1-g)$.
\end{enumerate}
According to the first implication, it is
without loss of generality to label the signal realizations such that $p(s_1) \geq p(s_2) \geq ... \geq p(s_m)$, and moreover, there exists $k \in \{1,2,...,m\}$ such that player $2$ plays $b^*$ for sure after observing $s_1,...,s_{k-1}$, and does not play $b^*$ otherwise.\footnote{Ignoring the possibility that player $2$ plays a mixed action following certain signal realizations is without loss of generality in proving the current theorem. This is because when player $2$ mixes between $n$ actions after one signal realization, we can split this signal realization into $n$ signal realizations with the same posterior belief, such that player $2$ plays a pure action following each of these signal realizations.} Therefore,
\begin{equation*}
    r (1-f(s^*|a^*)) = \sum_{i=1}^m \tau(s_i) p(s_i), \quad 1-r=\sum_{i=1}^m \tau(s_i) (1- p(s_i)), \quad \textrm{and } \sum_{i=k}^m \tau(s_i)=1-g.
\end{equation*}
Using the fact that $p(s_1) \geq p(s_2) \geq ... \geq p(s_m)$, we know that:
\begin{equation}\label{4.4}
    \frac{\sum_{i=1}^{k-1} \tau (s_i) p(s_i)}{\sum_{i=1}^{k-1} \tau (s_i) (1-p(s_i))}
    \geq \frac{r (1-f(s^*|a^*))}{1-r} \geq \frac{\sum_{i=k}^{m} \tau (s_i) p(s_i)}{\sum_{i=k}^{m} \tau (s_i) (1-p(s_i))}.
\end{equation}
As a result,
\begin{equation}\label{4.5}
 \sum_{i=k}^{m} \tau (s_i) p(s_i) \leq \frac{r (1-f(s^*|a^*))}{1-r f(s^*|a^*)} (1-g),
\end{equation}
and
\begin{equation}\label{4.6}
    \sum_{i=k}^{m} \tau (s_i) (1- p(s_i)) \geq \frac{1-r}{1-r f(s^*|a^*)} (1-g).
\end{equation}
Therefore,
\begin{equation}\label{4.7}
   \frac{ 1-g(\omega^c)}{1-g} \leq \frac{1-f(s^*|a^*)}{1-r f(s^*|a^*)},
\end{equation}
Using the second implication, namely, $r \leq 1-C (1-g)$, we have:
\begin{equation}\label{4.8}
    \frac{ 1-g(\omega^c)}{1-g} \leq \frac{1-f(s^*|a^*)}{1-f(s^*|a^*)+C f(s^*|a^*) (1-g)}.
\end{equation}
Similarly, the lower bound on the likelihood ratio with which $b^*$ occurs is given by:
\begin{equation}\label{4.9}
    \frac{g(\omega^c)}{g} \geq 1+ \frac{f(s^*|a^*) (1-g(h^t))}{g- r f(s^*|a^*)}
    \geq 1+ \frac{f(s^*|a^*) (1-g)}{g-f(s^*|a^*) (1-C(1-g))}
\end{equation}
Let $\beta(h^t) \in \Delta (B)$ be the distribution over player $2$'s action at $h^t$, and let $\beta(h^t,\omega^c) \in \Delta (B)$ be the distribution over player $2$'s action at $h^t$ conditional on player $1$ being the commitment type. Inequalities (\ref{4.8}) and (\ref{4.9}) imply the following lower bound on the KL divergence between $\beta(h^t)$ and $\beta(h^t,\omega^c)$:
\begin{equation}\label{4.10}
    d \Big( \beta(h^t)
    \Big| \beta(h^t,\omega^c)
    \Big) \leq \mathcal{L} \big(1-g(h^t)\big).
\end{equation}
The Pinsker's inequality implies that
$\mathcal{L}(\cdot)$ is of the magnitude $(1-g(h^t))^2$.

This lower bound on the KL divergence bounds the speed of learning at $h^t$ from below, as a function of the probability with which player $2$ at $h^t$ does not play $b^*$. This implies a lower bound on the speed of learning when player $2$ in the future observes $b^*$ in period $t$, \textit{given that he knew} that the probability with which player $2$ plays $b^*$ at $h^t$ is no more than $g(h^t)$. However, unlike models with unbounded memory, future player $2$'s information does not nest that of player $2$'s in period $t$. This is because future player $2$s may not observe $\{a_{t-K},...,a_{t-1}\}$, and hence, cannot interpret the meaning of $b_t$ in the same way as player $2$ in period $t$ does.

For every $s,t \in \mathbb{N}$ with $s>t$, I provide a lower bound on the informativeness of $b_t$ about player $1$'s type from the perspective of player $2$ who arrives in period $s$, as a function of the informativeness of $b_t$ (about player $1$'s type) from the perspective of player $2$ who arrives in period $t$. This together with (\ref{4.10}) establishes a lower bound on the informativeness of $b_t$ from the perspective of future player $2$s as a function of the probability with which $b^*$ is not being played. Applying the result in Gossner (2011), one obtains the commitment payoff theorem.

Let $\pi(h^t)$ be player $2$'s belief about $\omega$ at $h^t$ before observing the period $t$ signal $s_t$. By definition, $\pi(h^0)=\pi_0$.
For every strategy profile $\sigma$, let  $\mathcal{P}^{\sigma}$ be the probability measure over $\mathcal{H}$ induced by $\sigma$,
let $\mathcal{P}^{\sigma,\omega^c}$ be the probability measure induced by $\sigma$ conditional on player $1$ being the commitment type, and
let $\mathcal{P}^{\sigma,\omega^s}$ be the probability measure induced by $\sigma$ conditional on player $1$ being the strategic type.
One can the write the posterior likelihood ratio as the product of the likelihood ratio of the signals observed in each period:
\begin{equation*}
   \frac{ \pi(h^t)}{1-\pi(h^t)} \Big/ \frac{\pi_0}{1-\pi_0}
   \end{equation*}
 \begin{equation}\label{4.11}
   =\frac{\mathcal{P}^{\sigma,\omega^c}(b_0)}{\mathcal{P}^{\sigma,\omega^s}(b_0)} \cdot \frac{\mathcal{P}^{\sigma,\omega^c}(b_1|b_0)}{\mathcal{P}^{\sigma,\omega^s}(b_1|b_0)} \cdot ... \cdot
   \frac{\mathcal{P}^{\sigma,\omega^c}(b_{t-1}|b_{t-2},...,b_0)}{\mathcal{P}^{\sigma,\omega^s}(b_{t-1}|b_{t-2},...,b_0)}\cdot
   \frac{\mathcal{P}^{\sigma,\omega^c}(a_{t-K},...,a_{t-1}|b_t,b_{t-1},...,b_0)}{\mathcal{P}^{\sigma,\omega^s}(a_{t-K},...,a_{t-1}|b_t,b_{t-1},...,b_0)}
\end{equation}
Furthermore, for every $\epsilon>0$ and every $t$, we know that:
\begin{equation}\label{4.12}
    \mathcal{P}^{\sigma,\omega^c} \Big(\pi^{\sigma}(b_0,b_1,...b_{t-1}) < \epsilon \pi_0 \Big) \leq \epsilon \frac{1-\pi_0}{1-\pi_0 \epsilon},
\end{equation}
in which $\pi^{\sigma}(b_0,b_1,...b_{t-1}) \in \Delta (\Omega)$ is player $2$'s belief about player $1$'s type after observing $(b_0,...,b_{t-1})$ but before observing player $1$'s actions and $s_t$. For every $\epsilon>0$, let $\rho^*(\epsilon)$ be defined as:
\begin{equation}\label{4.13}
   \rho^*(\epsilon) \equiv \frac{\epsilon \pi_0}{1-C\epsilon}.
\end{equation}
Next, if $\pi^{\sigma}(b_0,b_1,...b_{t-1}) \geq \epsilon \pi_0$,
and player $2$ in period $t$ believes that $b_t=b^*$
occurs with probability less than $1-\epsilon$
after observing $(a_{t-K},...,a_{t-1})=(a^*,...,a^*)$,
then under probability measure $\mathcal{P}^{\sigma}$, the probability of
$\{a_{t-K},...,a_{t-1}\}=\{a^*,...,a^*\}$ conditional on $(b_0,...,b_{t-1})$ is at least  $\rho^*(\epsilon)$.

To see this, suppose towards a contradiction that the probability with which $(a_{t-K},...,a_{t-1})=(a^*,...,a^*)$ is strictly less than $\rho^*(\epsilon)$ conditional on $(b_0,...,b_{t-1})$.
According to (\ref{4.13}), after observing $(a_{t-K},...,a_{t-1})=(a^*,...,a^*)$ in period $t$ and given that
 $\pi^{\sigma}(b_0,b_1,...b_{t-1}) \geq \epsilon \pi_0$, $\pi(h^t)$
attaches probability strictly more than $1-C\epsilon$ to the commitment type. As a result, player $2$ in period $t$ believes that $a^*$ is played with probability at least $1-C\epsilon$ at $h^t$. This contradicts presumption that she plays $b^*$ with probability less than $1-\epsilon$.

Next, I study the believed distribution of $b_t$
from the perspective of player $2$ in period $s$ in the event that $\pi^{\sigma}(b_0,b_1,...b_{t-1}) \geq \epsilon \pi_0$.
Let $\mathcal{P}(\sigma,t,s) \in \Delta (\Delta (A^K))$ be player $2$'s signal structure in period $s (\geq t)$ about
$\{a_{t-K},...,a_{t-1}\}$ under equilibrium $\sigma$.
For every small enough $\eta>0$, given that $\mathcal{P}(\sigma,t)$ attaches probability at least $\rho^*(\epsilon)$ to
$\{a_{t-K},...,a_{t-1}\}=\{a^*,...,a^*\}$,
the probability with which $\mathcal{P}(\sigma,t,s)$ attaches to the event that
$\{a_{t-K},...,a_{t-1}\}=\{a^*,a^*,...,a^*\}$ occurs with probability less than $\eta \rho^*(\epsilon)$ conditional on $\{a_{t-K},...,a_{t-1}\}=\{a^*,a^*,...,a^*\}$ is bounded from above by:
\begin{equation}\label{4.14}
    \frac{\eta \rho^*(\epsilon) (1-\rho^*(\epsilon)) }{(1-\eta\rho^*(\epsilon)) \rho^*(\epsilon)}= \eta \frac{1-\rho^*(\epsilon)}{1-\rho^*(\epsilon)\eta}.
\end{equation}
Let $g(t|h^s)$ be player $2$'s belief about the probability with which $b^*$ is played in period $t$ when she observes $h^s$. Let
$g(t,\omega^c|h^s)$ be her belief about
the probability with which $b^*$ is played in period $t$ conditional on player $1$ being committed. The conclusions in (\ref{4.8}) and (\ref{4.9}) also apply in this setting, namely,
\begin{equation}\label{4.15}
    \frac{1-g(t,\omega^c|h^s)}{1-g(t|h^s)} \leq  \frac{1-f(s^*|a^*)}{1-f(s^*|a^*)+C f(s^*|a^*) (1-g(t|h^s))}
\end{equation}
and
\begin{equation}\label{4.16}
\frac{g(t,\omega^c|h^s)}{g(t|h^s)} \geq  1+ \frac{f(s^*|a^*) (1-g(t|h^s))}{g(t|h^s)-f(s^*|a^*) (1-C(1-g(t|h^s)))}
\end{equation}
Whenever player $2$ in period $s$ believes that
$\{a_{t-K},...,a_{t-1}\}=\{a^*,a^*,...,a^*\}$ occurs with probability more than $\eta \cdot \rho^* (\epsilon)$, we have:
\begin{equation}\label{4.17}
    g (t|h^s) \leq 1-  \epsilon \eta \rho^*.
\end{equation}
Applying (\ref{4.17}) to (\ref{4.15}) and (\ref{4.16}), we obtain a lower bound on the KL divergence between
$g(t,\omega^c|h^s)$ and $g(t|h^s)$. This is the lower bound on the speed with which player $2$ at $h^s$ will learn through $b_t=b^*$ about player $1$'s type, which applies to all events except for one that occurs with probability less than $\eta \frac{1-\rho^*}{1-\rho^*\eta}$.
Therefore, for every $\epsilon$ and $\pi_0$, there exists $\underline{\delta}$ such that when $\delta > \underline{\delta}$, the strategic player $1$'s payoff by playing $a^*$ in every period is at least:
\begin{equation}\label{4.18}
    \Big(1-\epsilon - \epsilon \frac{1-\pi_0}{1-\pi_0 \epsilon}\Big)
    u_1(a^*,b^*)
    +
    \Big(\epsilon  + \epsilon \frac{1-\pi_0}{1-\pi_0 \epsilon} \Big) \min_{a,b}u_1(a,b)
    -\epsilon.
\end{equation}
Taking $\epsilon \rightarrow 0$ and $\delta \rightarrow 1$, (\ref{4.18}) implies the commitment payoff theorem.

\subsection{Proof of Theorem 4: Statement 2}\label{subC.2}
Recall the definitions of $(a^*,b^*)$ and $(a',b')$ in the proof of Theorem \ref{Theorem1}.
I omit the trivial case in which $a^*=a'$, i.e., $a^*$ is player $1$'s action in his worst pure-strategy Nash Equilibrium.
I focus on the interesting case in which $a^* \neq a'$. Since $|A|=2$, we have $A=\{a^*,a'\}$.
Let
\begin{equation}\label{5.1}
  l^*(\mathbf{f}) \equiv  \max_{s \in S} \frac{f(s|a^*)}{f(s|a')}
\end{equation}
Consider the construction in the proof of Theorem \ref{Theorem1} with one modification: the overall probability with which player $1$ plays $a^*$ is:
\begin{equation}\label{5.2}
   \widehat{q} \equiv \frac{q^*}{q^*+(1-q^*)l^*(\mathbf{f})},
\end{equation}
and the probability with which he plays $a'$ is $1-\widehat{q}$. Let $\overline{\pi}_0=\widehat{q}^K$. Player $2$ has an incentive to play $b'$ in the reputation building phase, regardless of her observation of player $1$'s action in the past $K$ periods, and regardless of the signal she receives about player $1$'s action in the current period. When $K \geq 1$, the rest of the constructive proof follows from that of Theorem \ref{Theorem1}. When $K=0$, strategic-type player $1$ plays $a'$ at every history and player $2$ plays $b'$ at every history. Such a strategy profile is an equilibrium when the prior probability of commitment type is small enough.
\subsection{Proof of Theorem 4'}\label{subC.3}
For every $\alpha \in \Delta (A)$ and
$\beta: S \rightarrow \Delta (B)$, let $\pi(\alpha,\beta) \in \Delta (B)$ be the distribution over $b$ induced by $(\alpha,\beta)$.
When players' payoffs are monotone-supermodular and $\mathbf{f}$ is unboundedly informative about $a^*$ and satisfies MLRP, I establish the existence of $C>0$ such that for every $\alpha \in \Delta (A)$ with $a^* \in \textrm{supp}(\alpha)$, and every $\beta: S \rightarrow \Delta (B)$ that is player $2$'s stage-game best reply against $\alpha$ after observing the realization of $s$, if $\pi(\alpha,\beta)[b^*]<1-\varepsilon$, then
\begin{equation}\label{5.3}
    d\Big( \pi(\alpha,\beta) \Big\| \pi (a^*,\beta) \Big) > C \varepsilon^2.
\end{equation}
The rest of the proof follows from that of Theorem \ref{Theorem4} in Appendix \ref{subC.1} and is omitted to avoid repetition.

Let $\overline{A}$ be the set of actions that are strictly higher than $a^*$ and let $\underline{A}$ be the set of actions that are strictly lower
than $a^*$. Since $\mathbf{f}$ is unboundedly informative about $a^*$, there exists $s^* \in S$
such that $f(s^*|a)>0$ if and only if $a=a^*$. Let $\overline{S}$
be the set of signal realizations that are strictly higher than $s^*$ and let $\underline{S}$ be the set of signal realizations that are strictly lower than $s^*$. Since $\mathbf{f}$ satisfies MLRP, for every $s \in \overline{S}$, $f(s|a)>0$ only if $a \succeq a^*$, and for every $s \in \underline{S}$, $f(s|a)>0$ only if $a \preceq a^*$.

For every $\alpha \in \Delta (A)$, let $\alpha' \in \Delta (A)$ be the distribution over $A$ conditional on $a \neq a^*$.
If $\textrm{supp}(\alpha) \cap \overline{A} \neq \{\varnothing\}$, then
let $\overline{\alpha} \in \Delta (A)$ be the distribution over $A$ conditional on $a \in \textrm{supp}(\alpha) \cap \overline{A}$;
if $\textrm{supp}(\alpha) \cap \underline{A} \neq \{\varnothing\}$, then
let $\underline{\alpha} \in \Delta (A)$ be the distribution over $A$ conditional on $a \in \textrm{supp}(\alpha) \cap \underline{A}$.
By definition, there exists $\lambda \in [0,1]$
such that $\alpha'=\lambda \overline{\alpha}+(1-\lambda) \underline{\alpha}$.
When $\pi(\alpha,\beta)[b^*]<1-\varepsilon$, either
$\overline{\alpha}$ is well-defined and
$\pi(\overline{\alpha},\beta)[b^*]<1$, or
$\underline{\alpha}$ is well-defined and
$\pi(\underline{\alpha},\beta)[b^*]<1$, or both. Since players' payoffs are monotone-supermodular, and $\beta$ best replies against $\alpha$, there exist $\overline{s} \in \overline{S}$ and $\underline{s} \in \underline{S}$ such that $\beta(s)[b^*]=1$ only if $\overline{s} \succ s \succ \underline{s}$, and
$\beta(s)[b^*]>0$ only if $\overline{s} \succeq s \succeq \underline{s}$.

Suppose toward a contradiction that  $d\big( \pi(\alpha,\beta) \big\| \pi (a^*,\beta) \big) =0$, then $d\big( \pi(\alpha',\beta) \big\| \pi (a^*,\beta) \big) =0$, which suggests:
\begin{equation}\label{5.4}
    \lambda \Big(\beta(\overline{s})[b^*] \cdot f(\overline{s}|\overline{\alpha}) + \sum_{s \succ \overline{s}} f(s|\overline{\alpha}) \Big)
    =\beta(\overline{s})[b^*] \cdot f(\overline{s}|a^*)+ \sum_{s \succ \overline{s}} f(s|a^*)
\end{equation}
and
\begin{equation}\label{5.5}
  (1-  \lambda) \Big(\beta(\underline{s})[b^*] \cdot f(\underline{s}|\underline{\alpha}) + \sum_{s \prec \underline{s}} f(s|\underline{\alpha}) \Big)
    =\beta(\underline{s})[b^*]\cdot f(\underline{s}|a^*) + \sum_{s \prec \underline{s}} f(s|a^*).
\end{equation}
Since $f(s^*|a)>0$ if and only if $a=a^*$, (\ref{5.4}) and (\ref{5.5}) together imply that either
\begin{equation}\label{5.6}
    \lambda \Big( (1-\beta(\overline{s})[b^*]) \cdot f(\overline{s}|\overline{\alpha}) + \sum_{\overline{s} \succ s \succ s^*} f(s|\overline{\alpha}) \Big)
    > (1-\beta(\overline{s})[b^*]) \cdot f(\overline{s}|a^*) + \sum_{\overline{s} \succ s \succ s^*} f(s|a^*)
\end{equation}
or
\begin{equation}\label{5.7}
  (1-  \lambda) \Big((1-\beta(\underline{s})[b^*]) \cdot f(\underline{s}|\underline{\alpha}) + \sum_{\underline{s} \prec s \prec s^*} f(s|\underline{\alpha}) \Big)
    > (1-\beta(\underline{s})[b^*]) \cdot f(\underline{s}|a^*) + \sum_{\underline{s} \prec s \prec s^*} f(s|a^*).
\end{equation}
If (\ref{5.6}) is true, then (\ref{5.4}) and (\ref{5.6}) violate MLRP of $\mathbf{f}$.
If (\ref{5.7}) is true, then (\ref{5.5}) and (\ref{5.7}) violate MLRP of $\mathbf{f}$. Since the number of signal realizations is finite, there exists $C>0$ such that
$|| \pi(\alpha,\beta) - \pi (a^*,\beta)||> \frac{C}{2} \varepsilon$ whenever $\pi(\alpha,\beta)[b^*]<1-\varepsilon$. The Pinsker's inequality then implies that $ d\Big( \pi(\alpha,\beta) \Big\| \pi (a^*,\beta) \Big) > C \varepsilon^2$.

\subsection{Bounded Informativeness \textit{vs} Full Support}\label{subC.4}
I provide an example that explains why the full support condition in statement 2 of Theorem 4' cannot be replaced
by $\mathbf{f}$ being boundedly informative about $a^*$. Players' stage game payoffs are given by:
\begin{center}
\begin{tabular}{| c | c | c |}
\hline
  - & $b^*$ & $b'$  \\
  \hline
  $\overline{a}$ & $1,4$ & $-2,0$ \\
  \hline
  $a^*$ & $2,1$ & $-1,0$  \\
  \hline
    $\underline{a}$ & $3,-2$ & $0,0$ \\
  \hline
\end{tabular}
\end{center}
Let $S \equiv \{\overline{s},s^*,\underline{s}\}$, with $f(\overline{s}|\overline{a})=2/3$,
$f(s^*|\overline{a})=1/3$,  $f(\overline{s}|a^*)=1/3$,  $f(s^*|a^*)=2/3$,
and $f(\underline{s}|\underline{a})=1$. One can verify that players' stage-game payoffs are monotone-supermodular when player $1$'s actions are ranked according to $\overline{a} \succ a^* \succ \underline{a}$, and player $2$'s actions are ranked according to $b^*\succ b'$. When signal realizations are ranked according to $\overline{s} \succ s^* \succ \underline{s}$, $\mathbf{f}$ satisfies MLRP. It is easy to show that player $1$ can guarantee payoff $2$ in every Bayes Nash equilibrium. The reason is: if player $1$ plays $a^*$
in every period, player $2$ observes signal $s^*$ or $\overline{s}$, and has a strict incentive to play $b^*$. As a result,
player $1$'s payoff from playing $a^*$ in every period is at least $2$.

\end{spacing}

\end{document}